\documentclass[a4paper,11pt,oneside]{article}
\usepackage{latexsym}
\usepackage[english]{babel}
\usepackage{fancyhdr}
\usepackage[mathscr]{eucal}
\usepackage{amsmath}
\usepackage[left=1.5cm, right=1.5cm, top=2.3cm, bottom=2cm]{geometry}
\usepackage{mathrsfs}
\usepackage{amsthm}
\usepackage{amsfonts} 
\usepackage{latexsym}
\usepackage[english]{babel}
\usepackage{fancyhdr}
\usepackage[mathscr]{eucal}
\usepackage{amsmath}
\usepackage{mathrsfs}
\usepackage{amsthm}
\usepackage{amsfonts} 
\usepackage{amssymb} 
\usepackage{amscd}
\usepackage{bbm}
\usepackage{graphicx}
\usepackage{graphics}
\usepackage{geometry}
\usepackage{xcolor}
\usepackage{soul}

\def\Z{{\mathbb{Z}}}      \def\R{{\RR}}
\def\RR{{\mathbb{R}}}

\newtheorem{theorem}{Theorem}

\newtheorem{remark}{Remark} 
 
\makeatletter 

\title{Korteweg-de-Vries and Fermi-Pasta-Ulam-Tsingou: \\ asymptotic integrability of quasi unidirectional waves}  
\author{
Matteo Gallone\footnote{Universit\`a degli Studi di Milano, Dipartimento di Matematica ``F. Enriques'', Via Saldini 50, 20133 Milano, Italy; {\tt matteo.gallone@unimi.it}}, 
Antonio Ponno\footnote{Universit\`a degli Studi di Padova, Dipartimento di Matematica ``T. Levi-Civita'', Via Trieste 63, 35121 Padova, Italy;
{\tt ponno@math.unipd.it}}
and Bob Rink\footnote{Vrije Universiteit Amsterdam, Department of Mathematics, De Boelelaan 1111, 
1081 HV Amsterdam,  The Netherlands; {\tt b.w.rink@vu.nl} }
}

\begin{document}  \hyphenation{boun-da-ry mo-no-dro-my sin-gu-la-ri-ty ma-ni-fold ma-ni-folds re-fe-rence se-cond se-ve-ral dia-go-na-lised con-ti-nuous thres-hold re-sul-ting fi-nite-di-men-sio-nal ap-proxi-ma-tion pro-per-ties ri-go-rous De-no-ting}
\maketitle
\noindent 
\abstract{In this paper we construct a higher order expansion of the manifold of quasi unidirectional waves in the Fermi-Pasta-Ulam-Tsingou (FPUT) chain. We also approximate the dynamics on this manifold. As  perturbation parameter we use $h^2=1/n^2$, where 
$n$ is the number of particles of the chain. It is well known that the dynamics of quasi unidirectional waves is described to first order by the Korteweg-de Vries (KdV) equation. Here we show that the dynamics to second order is governed by a combination of the first two nontrivial
equations in the KdV hierarchy -- for any choice of parameters in the FPUT potential. On the other hand, we find that only if the parameters of the FPUT potential satisfy a condition, then a combination of the first three nontrivial equations in the KdV hierarchy determines the dynamics of quasi unidirectional waves to third order. The required condition is satisfied by the Toda chain. Our results suggest why the close-to-integrable behavior of the FPUT chain (the FPUT paradox)
persists on a time scale longer than explained by the KdV approximation, and also how a breakdown of integrability (detachment from the KdV hierarchy) \textcolor{black}{may be} responsible for the eventual thermalization of the system.}
\\ \mbox{} \\ 
\noindent 
{\bf Keywords:} Fermi-Pasta-Ulam-Tsingou, Korteweg-de Vries, Near-integrability, Normal forms.

\section{Introduction}

In the early 1950s, Fermi, Pasta, Ulam and Tsingou (FPUT) set up a series of numerical experiments, with the purpose to measure the time of approach to  statistical equilibrium in non-integrable Hamiltonian systems. Their experiments were motivated by the conviction of Fermi that generic, large size Hamiltonian systems are ergodic. FPUT studied a simple family of models: chains of particles  interacting with their nearest neighbors through a force of simple polynomial type. The unexpected outcome of their  numerical study \cite{FPU55}, namely the observed lack of ergodicity and mixing over the computational time then available, was named after them -- the FPUT problem, or FPUT paradox. Although a complete mathematical understanding of the FPUT problem is still lacking, some deep insights into  specific features have been obtained. It is not the aim of  this paper to review the history of the FPUT problem; the interested reader is referred to the existing reviews \cite{BCMM, Chaos,Sped1,Gall, Sped2}. 

Nowadays, we know that the FPUT paradox is a matter of quasi-integrability. By this we mean that the dynamics of the FPUT  chain over short timescales resembles that of a closeby integrable system, while the approach to statistical equilibrium, in the long run, is due to the perturbation. A Hamiltonian dynamical system is said to be integrable if, when viewed in the appropriate coordinate system, its solutions execute trivial, straight line motion. 
 Which integrable system should be considered ``closest'' to the FPUT chain, depends both on the model and on initial conditions. In this paper, we focus on the generic FPUT chain, or FPUT $\alpha$-model, in which the inter-particle forces display quadratic nonlinearities to leading order -- see equation \eqref{eq:PotentialEnergy} below. For this model, the first explanation of the FPUT paradox in terms of integrability goes back to the pioneering work of Zabusky and Kruskal \cite{ZK}. These authors show that the dynamics in the chain can be described, for short times and for smooth low energy initial conditions close to a unidirectional wave, by the Korteweg-de Vries (KdV) equation. This partial differential equation was later proved to be integrable \cite{GGKM,L68, MGK68,ZF71}. The  approach in \cite{ZK} was later enforced by Zakharov \cite{Z74}, who proved the integrability of the quadratic Boussinesq equation, which can be regarded a continuum approximation of the full FPUT system. 
On the other hand, it is also well understood now that the low energy short term dynamics of the FPUT chain follows that of the integrable Toda chain, for \emph{any} initial condition (even extracted in measure), as first pointed out in \cite{FFML, M74} and more recently in 
\cite{BCP13,BPP18, BP11,BP20, GMMP20}. Of course, when smooth initial data are considered, the two points of view connect, since the continuum limit of the Toda chain consists of two KdV equations \cite{BKP}. 

It was shown in recent studies \cite{BP06,BLP09, BP05} that the approximate description of the FPUT dynamics in terms of PDEs can be cast within the theory of normal forms. In particular it was proved that, in a certain regime of smooth low energy initial conditions, the leading order resonant normal form of the FPUT system consists of two KdV equations -- one describing almost-right traveling waves and the other almost-left traveling waves \cite{BP06}. This result agrees with the perspective sketched above: the dynamics of the FPUT chain is integrable in the short term. On the other hand, the numerically observed FPUT energy spectrum turns out to be stable over time scales much longer than that of the validity of the KdV approximation. This in turn suggests that the normal form of the FPUT system may actually be integrable even beyond the leading order KdV approximation.

We partially investigate this issue in the present paper. In fact, we will consider the FPUT chain with periodic boundary conditions, and we will introduce smooth functions interpolating the positions of the particles at any time. As small parameter we choose $h^2:=1/n^2$, where $n$ is the number of particles in the chain. 
In a forthcoming paper \cite{GPR2} we will consider this problem for arbitrary smooth initial data. In this paper, on the other hand, we shall investigate solutions that lie inside an invariant sub-manifold consisting of quasi unidirectional waves. We construct this invariant manifold to high perturbative orders. This method is among the standard 
techniques for finding invariant manifolds in hydrodynamics 
\cite{temam}, and was first introduced by 
Whitham \cite{whit}, and extended e.g. in \cite{dull}, to derive the KdV equation in the theory of shallow water waves.

The evolution equation that we initially find for the dynamics of quasi unidirectional waves, is not integrable beyond the KdV equation found at order $h^2$. To investigate the asymptotic integrability of the equation in more detail, we shall apply normal form transformations. The type of transformations that we consider were introduced by Hiraoka and Kodama \cite{HirKod}; in this paper we adapt them for continuous systems with periodic boundary conditions. Our results are the following. First of all, we find that to order $h^{4}$
 the dynamics of quasi unidirectional waves is \emph{always} governed by an integrable equation from the KdV hierarchy. In contrast, to order $h^6$ the dynamics is governed by a member of the KdV hierarchy only for particular values of the parameters   defining the nearest-neighbor interaction in the chain. The Toda chain is an example of such a system that to order $h^6$ possesses an integrable normal form. Thus,
\begin{center}
\emph{the dynamics of quasi unidirectional waves in generic FPUT chains is \\ integrable over the timescale corresponding to the second order normal form;\\ the breakdown of integrability generally takes place at third order}.
\end{center}
This agrees, qualitatively, with what is observed in numerical simulations. We also point out the recent work \cite{Bam20}, in which a somewhat similar result was proven in the theory of shallow water waves.

The remainder of this paper is organised as follows. In Section \ref{sec:formulation} we introduce the FPUT chain, as well as an exact continuous system interpolating it. In this section we also informally present the main result of this paper as Theorem \ref{thm:InformalMainResult}. In Section  \ref{sec:LongWaves} we provide higher order asymptotic expansions for the interpolating system, in the form of partial differential equations (PDEs). In Section \ref{sec:unidwaves} we construct and   approximate the manifold of quasi unidirectional  waves, and provide an asymptotic expansion of the dynamics on this manifold, see  Theorem \ref{thm:F}.
In Section \ref{sec:normform} we prove Theorem \ref{th:Main}, which provides the normal form for the dynamics on the manifold of quasi unidirectional waves. Theorem \ref{th:Main} is a direct consequence of Theorem \ref{thm:Kodama}, which is a spatially periodic version of a theorem by Hiraoka and Kodama. We prove Theorem \ref{thm:Kodama} in Section 
\ref{sec:KodamaProof}.

\section{Formulation of the main result} 
\label{sec:formulation}

In this section we give a brief description of the main result of this paper. 
We start by recalling that the periodic FPUT chain with $n$ particles  is the $n$-degrees of freedom Hamiltonian dynamical system with equations of motion
\begin{align} \label{eq:FPU-Lattice}
& \left\{ \begin{array}{ll}
\frac{d q_j}{dt} \;=&  p_j\, , \\ \frac{d p_j}{dt} \;=& W'(q_{j+1}-q_j) - W'(q_j-q_{j-1})\, ,
\end{array} \right. 
\end{align}
satisfying the periodic boundary conditions 
$$q_{j+n}=q_j,\ p_{j+n}=p_j \ \mbox{for all}\ j \in \Z\, .$$ The function 
$W$ in (\ref{eq:FPU-Lattice}) is a potential energy determining the interaction between neighboring particles in the chain. As usual, we assume that it admits the Taylor expansion
\begin{equation}\label{eq:PotentialEnergy}
W(z)=\frac{1}{2}z^2 + \frac{\alpha}{3}z^3+\frac{\beta}{4}z^4+\frac{\gamma}{5}z^5+ \mathcal O(z^6)\ . 
\end{equation}
The non-linearity in the  forces between neighboring particles is thus determined by the parameters $\alpha,\beta,\gamma$, etc.
We assume $\alpha \neq 0$ throughout this paper. As discussed above, a special role is played by the integrable Toda chain, \textcolor{black}{of which the potential energy is given by
\begin{equation}
\label{eq:TodaPotentialEnergy}
W(z) = W_{{\rm T}}(z) = \frac{1}{2} z^2 + \frac{\alpha}{3} z^3 + \frac{\alpha^2}{6}z^4 + \frac{\alpha^3}{15} z^5 +  \ldots\ ,
\end{equation}
(see Remark \ref{rem:Toda}) corresponding to a specific one-parameter family of potentials for which
\begin{equation}
\label{eq:Todapars}
\beta= \beta_{T} := \frac{2}{3} \alpha^2\, ,\ \ 
\gamma= \gamma_T:=\frac{1}{3} \alpha^3\, ,\ \mbox{etc.}
\end{equation}\noindent 
}
\noindent To transform the FPUT equations of motion (\ref{eq:FPU-Lattice}) into a continuous system of PDEs, we first replace (\ref{eq:FPU-Lattice}) by an exact evolution equation for an interpolating profile. To this end, 
let us denote $h:=1/n$ and write $\mathbb{T}:=\mathbb{R}/\mathbb{Z}$. Now consider a pair of smooth scalar functions  
$$(u, v)=(u,v)(x,t): \mathbb{T} \times \mathbb{R} \to \mathbb{R}^2$$ and assume that these functions satisfy the evolution equations
\begin{equation}
\label{eq:FPU-Continuous}
\left\{
\begin{array}{l}
u_t(x,t) =  v(x,t) \, ,\\
 v_t(x,t) =  h^{-3}\left[W'\left(h u(x + h, t)-h u(x,t)\right)- 
 W'\left(h u(x,t)-h u(x-h,t)  \right)\right]\, .
 \end{array}\right.
\end{equation}
We informally think of \eqref{eq:FPU-Continuous} as an ODE on the space $C^{\infty}(\mathbb T, \mathbb R^2)$ of pairs of smooth scalar functions on $\mathbb T$. We correspondingly think of solutions to \eqref{eq:FPU-Continuous} as curves $t\mapsto (u,v)(\cdot)(t)$ in $C^{\infty}(\mathbb T, \mathbb R^2)$. 

Given any solution of (\ref{eq:FPU-Continuous}), one can define, for $j\in \mathbb Z$, the functions 
\begin{equation} 
\label{eq:interpolation}
\left\{
\begin{array}{l} q_j(t):= h\, u(h j,h t) \, , \\ 
p_j(t) := h^2\, v(h j,h t)\, .
\end{array}\right.
\end{equation}
It is not hard to check that these $q_j(t), p_j(t)$ then form solutions of the FPUT equations of motion (\ref{eq:FPU-Lattice}). This motivates us to study \eqref{eq:FPU-Continuous} instead of (\ref{eq:FPU-Lattice}), the advantage being that \eqref{eq:FPU-Continuous} is defined on the same phase space for each value of $h$; in fact, we may simply think of $h$ as a continuous small parameter.
As long as $u$ is a smooth function of $x$, we can Taylor expand $u(x\pm h) = u(x)\pm hu_x(x) + \frac{h^2}{2} u_{xx}(x) + \ldots$, and write \eqref{eq:FPU-Continuous} as the perturbed Boussinesq equation

\begin{equation}
\label{eq:FPU-Continuousexpanded11}
\left\{
\begin{array}{l}
u_t =  v \, ,\\
 v_t =  u_{xx} + h^{2} \left( \frac{1}{12}  u_{xxxx}  + 2 \alpha  u_x  u_{xx} \right) + \mathcal{O}(h^{4})  \, .
 \end{array}\right.
\end{equation}
\textcolor{black}{
Here we have (with abuse of notation) denoted by $u_t$ and $v_t$ the derivatives of $u$ and $v$ with respect to their second argument. In other words, we have rescaled time $t\mapsto h t$. 
}
Next, instead of studying \eqref{eq:FPU-Continuous} or  \eqref{eq:FPU-Continuousexpanded11} directly, we  introduce a change of variables $(u,v)\mapsto(U,V)$ in $C^{\infty}(\mathbb T, \mathbb R^2)$, that 
maps the functions $(u, v)(x)$ to   
the  \emph{discrete Riemann invariants} $(U,V)(x)$ defined by
\begin{equation} 
\label{eq:Riemanninv}
U := 2\alpha(D_{h}u + v) \, ,  \ 
 V := 2\alpha(D_{h}u - v)\ .
\end{equation}
Here $D_h$ denotes the finite difference operator 
\begin{equation}
\label{eq:Dhu}
(D_{h}u)(x):=\frac{u(x+h/2)-u(x-h/2)}{h}\ .
\end{equation}
This $D_h$ is a discrete approximation of the derivative $\partial_x$, because $D_hu = u_x + \frac{h^2}{24} u_{3x} + \mathcal O(h^5)$ if $u: \mathbb T \to \mathbb R$ is sufficiently smooth - again by Taylor's theorem. 
The change of variables $(u,v)\mapsto (U,V)$ transforms the continuum equations
(\ref{eq:FPU-Continuous}) into a system of the form
\begin{equation}
\label{eq:exactriemanninv1}
\left\{
\begin{array}{l}
U_t =  \mathscr{F}(U,V,h)\, ,     \\ 
V_t = -\mathscr{F}(V,U,h)\, .
\end{array}\right.
\end{equation}
 We again think of \eqref{eq:exactriemanninv1} as an ODE on $C^{\infty}(\mathbb T, \mathbb R^2)$. An exact expression for $\mathscr{F}(U,V,h)$ will be given in Section \ref{sec:LongWaves}. Here it suffices to know that $\mathscr{F}(U,V,h)$ admits an expansion
$$\mathscr{F}(U,V,h) =   U_x +  \frac{h^{2}}{24}\left( U_{xxx} + 6 UU_x + 6UV_x + 6U_xV + 6VV_x \right)  + \mathcal O(h^{4})\, ,$$
as long as $(U,V)$ are sufficiently smooth functions of $x$. 

By setting $V\equiv 0$ in the evolution equation for $U$ in \eqref{eq:exactriemanninv1} we recover the KdV equation 
$$U_t = \mathscr{F}(U,0,h) = U_x + \frac{h^{2}}{24}\left(  U_{xxx} + 6 U U_x \right) + \mathcal O (h^4)\, .$$ 
Nevertheless, it should be clear that if $V(x) \equiv 0$ while $U(x)\not \equiv 0$, then 
$$V_t  = - \mathscr{F}(V, U ,h) = -\mathscr{F}(0, U ,h) = -\frac{1}{4} h^2 UU_x \not \equiv 0\, .$$ 
This means that, already to order $h^2$, the subspace 
$$\{(U,V) \in C^{\infty}(\mathbb{T}, \mathbb R^2) \, | \, V\equiv 0\} \subset C^{\infty}(\mathbb{T}, \mathbb R^2)$$ 
of ``unidirectional waves'' is not invariant under the flow of  \eqref{eq:exactriemanninv1}. 
This may cast some doubt on the validity of the above derivation of the KdV approximation, but we will show that nearby the subspace of unidirectional waves one can formally construct a submanifold of quasi unidirectional waves of the form
$$ \{(U,V) \in C^{\infty}(\mathbb{T}, \mathbb R^2) \, | \, V = c(U,h) = \mathcal{O}(h^2)  \} \subset C^{\infty}(\mathbb{T}, \mathbb R^2)\, .$$ 
This submanifold is invariant under \eqref{eq:exactriemanninv1} to high asymptotic orders. 
Moreover, the dynamics on the manifold of quasi unidirectional waves
 is indeed governed to order $h^2$ by the KdV equation. 
 
 More interestingly, we shall also derive higher order equations for the evolution 
of quasi unidirectional waves. It turns out that, after a suitable change of variables, this evolution is determined to a large extent by the ``higher order versions'' of the KdV equation, i.e., by the integrable KdV hierarchy. 
More precisely, the main results of this paper are Theorems \ref{thm:F} and \ref{th:Main}, which can be summarised as follows. 
\begin{theorem}
\label{thm:InformalMainResult}
Inside $C^{\infty}(\mathbb{T}, \mathbb{R}^2)$ there is a formal invariant manifold for the  dynamics \eqref{eq:exactriemanninv1} of the discrete Riemann invariants, consisting of quasi unidirectional waves, and defined by a slaving relation  of the form 
$$V = c(U,h) = h^2 c_2(U) + h^4 c_4(U) + \mathcal{O}(h^6) \, . $$
There also exists a formal near-identity transformation in $C^{\infty}(\mathbb{T}, \mathbb{R})$ of the form
$$U\mapsto U+h^2\widetilde G_2(U) + h^4\widetilde G_4(U)+ h^{6}\widetilde G_6(U,t) + \mathcal{O}(h^8)\, ,$$ 
bringing the dynamics on the manifold of quasi unidirectional waves into the form
\begin{align}
\label{eq:MainResultIntroduction}
\begin{split}
 U_t  &  = \hspace{5mm}  \mathcal{C}_1 (U,h)\ K_1(U) \, + \\ 
\ &  {}^{}+ h^{2}\ \mathcal{C}_3(U,h) \ K_3(U)\, + \\
\ &  {}^{}+ h^{4}\  \mathcal{C}_5(U,h)\ K_5(U)\, + \\
\ &  {}^{}+ h^{6}\  \mathcal{C}_7(U, h) \left[K_7(U)+R(U)\right] + \mathcal O(h^{8})\, .
\end{split}
\end{align}
Here, $K_1(U), K_3(U), K_5(U), K_7(U)$ are the first four commuting vector fields in the KdV
hierarchy, and the scalars $\mathcal{C}_1(U,h),\mathcal{C}_3(U,h), \mathcal{C}_5(U,h), \mathcal{C}_7(U,h)$ are constants of motion of the KdV hierarchy (see also Remark \ref{rmk:hierarchy}).
The term $R(U)$ can be chosen equal to zero in case 
\begin{equation}
\label{eq:condmainres}
14\alpha^3  - 27\beta \alpha + 12\gamma =0\ .
\end{equation} 
\end{theorem}
\noindent 
Theorem \ref{thm:InformalMainResult} shows that the evolution of quasi unidirectional waves is always integrable to order $h^4$, and is integrable to order $h^6$ if the relation \eqref{eq:condmainres} holds. This relation holds in particular when $\beta = \beta_T =2 \alpha^2/3$ and $\gamma = \gamma_T = \alpha^3/3$. The one-parameter family of Toda chains is thus part of a co-dimension one family of FPUT chains whose dynamics is ``more integrable'' than that of generic FPUT chains.
\begin{remark}\label{rmk:hierarchy}
 The first four commuting vector fields in the KdV hierarchy are given by 
\begin{equation}
\label{eq:KdV-hierarchy}
\begin{split}
K_1(U) \; &= \; U_x \, , \\
K_3(U) \; &= \;U_{3x}+6\, UU_x  \, , \\
K_5(U) \; &= \; U_{5x} + 20\, U_xU_{2x} +  10\, UU_{3x} + 30\, U^2U_x \, , \\
K_7(U) \; &= \; U_{7x} + 70\, U_{2x}U_{3x} + 42\, U_xU_{4x} + 14\, UU_{5x}+  \\ 
&\qquad \!\!\!\!+ 70\, U_x^3 + 280\, UU_xU_{2x} + 70\, U^2U_{3x} + 
140\, U^3U_x\ .
\end{split}
\end{equation}
See for instance \cite{MGK68}.  
 We used the short-hand notation
$U_{mx} := \partial_{mx}U =\partial_x^m U$ for the $m$-derivative of $U$. 

The constants of motion $\mathcal C_1(U, h), \mathcal C_3(U, h), \mathcal C_5(U, h), \mathcal C_7(U, h)$ in Theorem \ref{th:Main} turn out to be functions of the first three integrals of the KdV hierarchy,
\begin{equation}
\label{eq:KdV-integrals}
  \int_{\mathbb{T}} U \, dx\, , \
  \int_{\mathbb{T}} U^2\, dx \ \mbox{and}  \ 
  \int_{\mathbb{T}}  U_x^2 -2 U^3 \, dx  \, .
\end{equation}
\end{remark}

\begin{remark}\label{re:energy}
If \eqref{eq:interpolation} holds for smooth functions $u, v:\mathbb T \to \mathbb R$,  then the total energy 
$$E = \sum_{j=1}^n \frac{1}{2} {p_j^2 } + W(q_{j+1} - q_j) $$
of the FPUT system \eqref{eq:FPU-Lattice} satisfies
$$E = h^3 \int_{\mathbb{T}} \frac{1}{2} v(x)^2 + \frac{1}{2} u_x(x)^2 \, dx + \mathcal{O}(h^4) \, . $$
In turn this implies that the ``specific energy'' (the energy per particle) $\frac{1}{n}E$ is of the order $h^4 = 1/n^4$. This is far from the thermodynamic limit in which $\frac{1}{n}E$  would be of order $1$. However, \eqref{eq:interpolation} is the unique scaling limit in which dispersion and nonlinearity are of the same order,  thus leading at lowest order to a KdV equation for the evolution of quasi unidirectional waves. It is also the scaling limit in which the Toda chain can be smoothly connected to the KdV equation \cite{BKP} and topologically to the harmonic chain \cite{BM}. 
\end{remark}
{\color{black}
\begin{remark}\label{timescaleremark}
Theorem \ref{thm:InformalMainResult} states that the evolution equation for unidirectional waves in the FPUT-chain is asymptotically integrable up to and including order $h^5$. This suggests that non-integrable effects can only be observed over timescales of the order at least $1/h^6$ or, when the time rescaling 
$t\to ht$ discussed directly below formula (\ref{eq:FPU-Continuousexpanded11}) is accounted for, timescales of the order $1/h^7$. This leads to the conjecture that 
 the timescale of  thermalization of the FPUT-chain must be at least $1/h^7 \sim (E/n)^{-7/4}$. This is indeed shorter than, and thus compatible with, the thermalization timescale $(E/n)^{-9/4}$ that was numerically found for generic FPUT-chains in \cite{BP11}.  Theorem \ref{thm:InformalMainResult}  also suggests that the thermalization timescale is considerably longer for FPUT-chains satisfying $14\alpha^3 - 27\beta \alpha + 12 \gamma = 0$. However, we do not know of any rigorous way to employ Theorem \ref{thm:InformalMainResult} to prove such lower bounds for the thermalization timescale. 
\end{remark}
}
 
{\color{black}
\begin{remark}
The authors of \cite{HK} and \cite{Bob} compute and analyse the Birkhoff-Gustavson normal form of the finite FPUT-chain. This can be done, for example, by applying in \eqref{eq:FPU-Lattice} the rescaling 
\begin{equation} \nonumber 
  q_j(t):= \epsilon\, u_j (t) \ \mbox{and} \
p_j(t) := \epsilon\, v_j (t)\, \ \mbox{for} \ 1 \leq j \leq n\ ,
\end{equation}
while leaving time unaffected. The small parameter $\epsilon$ is here assumed independent of the number of particles $n$, and hence thought of as a fixed number. Under this rescaling, we have 
$E = \mathcal{O}(\epsilon^2)$. The resulting evolution equations are  of the form 
$$\frac{d}{dt}(u,v) = F_1(u,v) + \epsilon F_2(u,v) + \epsilon^2 F_3(u,v) +  \mathcal{O}(\epsilon^3)\ \mbox{for} \ (u,v)\in \mathbb{R}^{2n}\, .$$
Here, $F_1$ denotes the linear part of the equations, $F_2$ the quadratic part of the equations before rescaling, $F_3$ the cubic part of the equations before rescaling, etc. It was shown in \cite{HK, Bob} that, for every finite $n$, appropriate symplectic transformations can bring these equations  into a normal form $$\frac{d}{dt}(u,v) = F_1(u,v) + 
\epsilon^2\, \overline{F}_3(u,v) + \mathcal{O}(\epsilon^3)\ ,$$
where $\overline{F}_3$ is a vector field of polynomial degree three. 
It was proved by the same authors that this normal form is integrable to order $\epsilon^2$, i.e., if the $\mathcal{O}(\epsilon^3)$ term is ignored. By the same argument as in Remark \ref{timescaleremark}, this  suggests that the thermalization time is at least 
$\epsilon^{-3} \sim E^{-3/2}$. 

For finite chains with $n=16, 32$ and $64$ particles, this result was improved in \cite{On15}, using the so-called weak turbulence formalism, to a near-integrability result of order $\epsilon^3$, and hence an estimated thermalisation time of at least $\epsilon^{-4} \sim E^{-2}$. 
A statistical argument was then used to explain the numerical evidence of a much longer thermalization time, namely $\epsilon^{-8}\sim E^{-4}$.

However, all these results strongly rely on the assumption that $n$ is finite, as their proofs exploit the nonresonance between the eigenvalues of the finite FPUT-chain \eqref{eq:FPU-Lattice}. The domain of validity of the normal form / weak turbulence approximation may shrink dramatically as $n$ grows. This makes it highly nontrivial to draw conclusions from the results in \cite{HK, On15, Bob} for FPUT-chains with $n$ growing to infinity. It also makes it difficult to compare the results from \cite{HK, On15, Bob} with the results of the present paper.  
\end{remark}
}
 
\section{PDE approximations of the FPUT dynamics}
\label{sec:LongWaves}

In this section we provide more details on the derivation of the evolution equations \eqref{eq:exactriemanninv1} for the discrete Riemann invariants $(U,V)$ defined in \eqref{eq:Riemanninv}. Our starting point is the evolution equation for the interpolating profiles $(u,v)$ given in \eqref{eq:FPU-Continuous}:
\begin{equation}\nonumber 
\left\{
\begin{array}{l}
u_t(x,t) =  v(x,t)\, , \\
 v_t(x,t) =  h^{-3}\left[W'\left(h u(x + h, t)-h u(x,t)\right)- 
 W'\left(h u(x,t)-h u(x-h,t)  \right)\right] \, .
 \end{array}\right.
\end{equation}
Assuming that $u(x,t)$ and $v(x,t)$ are smooth functions of $x$, one may Taylor expand $u(x\pm h, t)$ with respect to $h$ in the right hand side of (\ref{eq:FPU-Continuous}). This  yields a perturbed nonlinear wave equation which we present  here to high order:
\begin{equation}
\label{eq:FPU-Continuousexpanded}
\left\{ \begin{split}
u_t & = v \, , \\ 
v_{t} & =  u_{xx} + h^{2} \left( \frac{1}{12}  u_{4x}  + 2 \alpha  u_x  u_{xx} \right) + h^{4} \left(\frac{1}{360}  u_{6x} + \frac{\alpha}{3}  u_{xx}  u_{3x} + \frac{\alpha}{6}  u_x  u_{4x} +  3 \beta  u_x^2  u_{xx}\right)
\\
& +   h^{6}\left( \frac{1}{20160}u_{8x} +  \frac{\alpha}{36} u_{3x}u_{4x} + \frac{\alpha}{60}u_{2x}u_{5x} + \frac{\alpha}{180}u_xu_{6x} +\right. \\
& +\left. \frac{\beta}{4}u_{2x}^3 + \beta u_x u_{xx}u_{3x}  + \frac{\beta}{4}u_x^2u_{4x} + 4 \gamma u_x^3 u_{2x}\right) +  h^{8}R(u,h)\ .
\end{split} \right.
\end{equation}
Note however that the  assumption that $u$ and $v$ are smooth  has no a priori justification: it is not guaranteed that (\ref{eq:FPU-Continuous}) admits nontrivial initial data leading to solutions that remain smooth over long times. In contrast to this, we recall here that the dynamics of the KdV equation (and of its hierarchy) on $\mathbb{T}$ preserves the Sobolev smoothness and analyticity of initial data for all times
\cite{KP09}. 
\begin{remark}
Equation (\ref{eq:FPU-Continuousexpanded}) reveals that the FPUT equations (\ref{eq:FPU-Continuous}) may be thought of as a weakly dispersive and weakly nonlinear perturbation of the wave equation. The perturbation to  order $h^2$ is the so-called Boussinesq equation. This equation was proven to be integrable by Zakharov \cite{Z74} in order to explain the FPUT paradox. 
\end{remark}
\begin{remark}
Using Taylor's theorem, one may determine explicit bounds for the remainder $h^{8}R(u,h)$ in terms of $u$. For example, one may easily get a conditional estimate of the form: for every integer $k\geq 0$, there exists a constant $C_k$ independent of $u$ and $h$, such that 
$$\|R(u,h)\|_{C^k(\mathbb{T})} 
\leq C_k \|u\|_{C^{k+10}(\mathbb{T})}$$ 
as long as $\|u\|_{C^0(\mathbb{T})}<1$ and $|h|<1$. Similar estimates can be obtained for Sobolev norms. We will not pursue such explicit bounds any further. Instead, from now on we shall simply write
$\mathcal{O}(h^m)$ for expressions of the form $h^mR(u, h)$. 
\end{remark}
\noindent 
For the exact wave equation $u_{t}=v$, $v_t =u_{xx}$ it is common to make the change of variables $(u,v)\mapsto (U,V)$ defined by  $U:= 2a(u_x +v)$ and $V:= 2a(u_x - v)$, where $a\neq0$ is any constant. These functions $U$ and $V$ are   called {\it Riemann invariants}, and their evolution is determined by the equations of motion 
\begin{equation}\nonumber
\left\{
\begin{split}
      U_t & = U_x\, , \\ V_t & = - V_x\, . 
\end{split}
\right.
\end{equation}
of which the solutions 
$$U(x,t)=U_0(x+t)\, ,\ V(x,t)=V_0(x-t)$$ 
are unidirectional traveling waves. This motivates our 
definition (\ref{eq:Riemanninv}) of the discrete Riemann invariants for (\ref{eq:FPU-Continuous}), with the particular choice $a=\alpha$ and $u_x$ replaced by 
$D_hu$:  
\begin{equation} \nonumber
U := 2\alpha(D_{h}u + v) \, , \ \ 
 V := 2\alpha(D_{h}u - v)\,  .
\end{equation}
We stress that our definition of the discrete Riemann invariants -- choosing $a=\alpha$ and using $D_hu$ instead of $u_x$ -- is such that the structure of their evolution equations is relatively simple. To see this, note that \eqref{eq:FPU-Continuous} can also be written as
\begin{equation}
\label{eq:FPU-contsint}
\left\{
\begin{array}{l}
u_t=v\, , \\
v_t=h^{-2}D_hW'(h^2D_hu)\, .
\end{array}\right.
\end{equation}
\begin{remark}
This follows from the identity
\begin{equation}
\label{eq:Dhid}
F(G(x+h)-G(x))-F(G(x)-G(x-h))= hD_hF(hD_hG(x))\ ,
\end{equation}
which holds for any pair of functions $F$ and $G$, applied to $F=W'$ and $G=u$.
\end{remark}
\noindent
Taking the time derivative of $U$ and $V$ as defined in 
\eqref{eq:Riemanninv}, and using \eqref{eq:FPU-contsint}, we find 
\begin{equation}
\label{eq:exactriemanninv}
\left\{
\begin{array}{l}
U_t =  \mathscr{F}(U,V,h) :=  D_{h} \left[ U + f(U+V,h) \right] \, , \\ 
V_t = -\mathscr{F}(V,U,h) := - D_{h} \left[V + f(U+V,h) \right] \, .
\end{array}\right.
\end{equation}
Here
\begin{equation}
\label{eq:f}
f(z,h):= 2\alpha h^{-2}\left[ W'\left(\frac{h^2z}{4\alpha}\right) - \frac{h^2z}{4\alpha}\right] = 
\frac{h^2}{8}z^2 + \frac{\beta h^4}{32\alpha^2} z^3 +  
\frac{\gamma h^6}{128\alpha^3}z^4 +  \mathcal O(h^8)
\end{equation}
is the (rescaled) nonlinear part of the inter-particle force. Using that
\begin{equation}
\label{eq:Dhuexpand}
D_{h} = \partial_x + \frac{h^2}{24}\partial_{3x}+ \frac{h^4}{1920}\partial_{5x} + \frac{h^6}{322560}\partial_{7x} + \mathcal O(h^8)\, ,
\end{equation}
one may expand 
$\mathscr{F}(U,V,h)$ in
\eqref{eq:exactriemanninv} to high order. This gives (when $U,V$ are smooth):
\begin{equation}
\label{eq:EvolutionU}
\begin{split}
U_t = \mathscr{F}(U,V,h) \;&=\;  U_x +  h^{2}\left( \frac{1}{24}U_{3x} + \frac{1}{8} (U+V)^2_x \right)  \\ 
 & + h^{4}\left( \frac{1}{1920}U_{5x} + \frac{1}{192}(U+V)^2_{3x}  + 
 \frac{\beta}{32\alpha^2}(U+V)^3_x\right) \\   
& + h^{6}\left( \frac{1}{322560}U_{7x} + \frac{1}{15360} (U+V)^2_{5x}  + \frac{\beta}{768\alpha^2} (U+V)^3_{3x} + 
\frac{\gamma}{128\alpha^3}(U+V)^4_x \right) \\
& +\mathcal O(h^{8})\ .
\end{split}
\end{equation}
\begin{remark}
Here and in the the sequel, it is understood that
 $$F^n_{mx}:=(F^n)_{mx}$$ 
is short hand notation for the $m$-th derivative of the $n$-th power of the function $F$ (and not the $n$-th power of the $m$-th derivative).
\end{remark}
\begin{remark}
Note that the evolution equation for $V$ in \eqref{eq:exactriemanninv} can be obtained from that for $U$ in
\eqref{eq:exactriemanninv} by exchanging the roles of $U$ and $V$ and adding a minus sign. This expresses the fact that \eqref{eq:exactriemanninv}  possesses the time reversal symmetry 
$$(U, V, t) \mapsto (V,U, -t)\, , $$  
corresponding in turn to the time reversal symmetry 
$$(u,v,t)\mapsto (u, -v, -t)$$
of the original continuous FPUT system \eqref{eq:FPU-Continuous}. It is therefore not restrictive, in what follows, to focus our analysis on the evolution of $U$, since analogous results automatically hold for $V$.
\end{remark} 

\section{Quasi unidirectional waves}
\label{sec:unidwaves}

The subspace 
$$\{(U,V)\in C^{\infty}(\mathbb T, \mathbb{R}^2)\, |\, V \equiv 0\} \subset C^{\infty}(\mathbb T, \mathbb{R}^2)$$
is invariant under the evolution of the wave equation 
$$U_t=U_x\, , \ V_t=-V_x\, ,$$
and consists entirely of unidirectional traveling waves $$U(x,t)=U_0(x+t), V(x,t)=0\, .$$
The evolution of the discrete Riemann invariants of the FPUT chain is accurately described by the wave equation to order $h^0$ because 
$U_t = \mathscr F(U,V,h) = U_x + \mathcal{O}(h^2)$ and $V_t = -\mathscr F(V, U,h) = -V_x + \mathcal{O}(h^2)$. Nevertheless, the subspace of unidirectional waves is already not invariant anymore to order $h^2$ -- see Section \ref{sec:formulation}.

It is therefore quite natural to search for an (at least asymptotically) invariant manifold for the dynamics of \eqref{eq:exactriemanninv} lying close to the subspace of unidirectional waves. We will try to find such a submanifold in the form of the graph of a function $c=c(U,h)=\mathcal O(h^2)$ over the subspace of unidirectional waves:
$$\{(U,V)\in C^{\infty}(\mathbb T, \mathbb{R}^2)\, |\, V =c(U,h) = \mathcal O(h^2) \}\, .$$
\begin{center}
\emph{We think of this graph 
$\{(U,V) \, |\, V=c(U,h)\}$ as a manifold of quasi unidirectional waves}.
\end{center}
\begin{remark}
By Definition \eqref{eq:Riemanninv} we have that $V=\textcolor{black}{2} \alpha(u_x - v) + \textcolor{black}{\mathcal{O}(h^2)}$. The assumption that 
$V=\mathcal{O}(h^2)$ thus means that $u_x=v+\mathcal{O}(h^2)$.  This in turn implies,
by the equations of motion \eqref{eq:FPU-Continuous}, that
$u_t=u_x+\mathcal{O}(h^2)$ and $v_t=v_x+\mathcal{O}(h^2)$. Thus we recover an approximate left traveling wave in the original variables.
\end{remark}
\noindent 
Inserting the ansatz 
\begin{equation}
\label{eq:DefC}
V = c(U,h)\ 
\end{equation} 
into
system \eqref{eq:exactriemanninv} yields the following two equations:
\begin{equation}
\label{eq:invman}
\left\{
\begin{split}
U_t = & \mathscr{F}(U,c(U,h),h)\, , \\ 
c(U,h)_t  = &-\mathscr{F}(c(U,h),U,h)\, .\\
\end{split}\right.
\end{equation}
The first of these equations states that solutions $(U,V) = (U, c(U,h))$ inside the manifold of quasi unidirectional waves are  governed by a closed evolution equation
\begin{equation} \label{eq:quasiUevolution}
    U_t=\mathcal{F}(U,h):=\mathscr{F}(U,c(U,h),h)\, .
\end{equation}
Once a solution $U$ of \eqref{eq:quasiUevolution} has been found, the solution $V$ is completely determined by  \eqref{eq:DefC}. 
We therefore think of the functional relation \eqref{eq:DefC} between $U$ and $V$ as a slaving relation: the free variable $U$ completely determines the behaviour of the slave variable $V$.

Using the chain rule at its left hand side, the second equation in \eqref{eq:invman}  can be rewritten as
\begin{equation}
\label{eq:invman2}
\begin{split}
c'(U,h)\mathscr{F}(U,c(U,h),h)=-\mathscr{F}(c(U,h),U,h)\, .\\
\end{split}
\end{equation}
This equation can be thought of as an invariance equation. It states that a certain relation must hold between the $U$-component and the $V$-component of the vector field in \eqref{eq:exactriemanninv} restricted to the manifold of quasi unidirectional waves. This relation guarantees the manifold to be invariant.

\begin{remark}\label{rmk:gateaux}
We denote by $F'(U)$ the Gateaux derivative (or directional derivative) of an operator $F$ at $U\in C^{\infty}(\mathbb T, \mathbb R)$. In other words, given a smooth increment function $H\in C^{\infty}(\mathbb T, \mathbb R)$, we have $$F'(U) H := \lim_{\varepsilon \to 0} 
\frac{F(U+ \varepsilon H) - F(U)}{\varepsilon}\, .$$
The operators $F=F(U)$ that we consider in this paper generally are maps from $C^{\infty}(\mathbb T, \mathbb R)$ to $C^{\infty}(\mathbb T, \mathbb R)$ and may depend for example on $U, U_x, U_{xx}$, etc. but also on certain  averages or even primitives of $U$.
\end{remark}
\noindent The following result gives a formula for
$c(U,h)$ to order $h^4$, and for the vector field 
$\mathcal{F}(U,h)$ to order $h^6$. We use the notation  $\langle F\rangle := \int_{\mathbb T} F(x)\, dx$ for $F=F(x):\mathbb T \to \mathbb R$, see also Remark \ref{rmk:averages}. 
\begin{theorem}\label{thm:F}
A formal invariant manifold of quasi unidirectional waves, defined by the invariance equation \eqref{eq:invman2}, is given to order $h^4$ as the graph of the function
\begin{equation}
\label{eq:cU}
\begin{split}
c(U, h)  &:= h^2 \left\{ \frac{1}{16}\left( \langle U^2 \rangle  -U^2 \right)\right\}
+ \\
& + h^4 \left\{ \left(\frac{5}{384} - \frac{\beta}{64\alpha^2} \right) 
\left( U^3 -\langle U^3\rangle \right)\right.  - \left.\frac{1}{128}\langle U^2\rangle (U  -\langle U \rangle)
-\frac{1}{256}\left[ (U_x)^2- \langle (U_x)^2\rangle\right] \right\}+\\ 
& + 
\mathcal{O}(h^6)\, .
\end{split}
\end{equation} 
The vector field defined in \eqref{eq:quasiUevolution} that determines the evolution of quasi unidirectional waves,  is given by
\begin{equation}\label{eq:Fu}
\begin{split}
\mathcal{F}(U,h) &= U_x + \frac{h^{2}}{24}\left\{U_{3x} + 6 U U_x \right\} +  
\\
& + \frac{h^{4}}{1920} \left\{U_{5x} + 60 U_xU_{2x} +   20 UU_{3x}  + 90 \left( \frac{2\beta}{\alpha^2} - 1 \right) U^2U_x + 30 \langle U^2\rangle U_x \right\} + \\
&  +
\frac{h^{6}}{322560}\left\{U_{7x} + 420 U_{2x}U_{3x} + 210 U_xU_{4x}+ 42 UU_{5x} +315 \left( \frac{8\beta}{\alpha^2} - 5 \right) (U_x)^3 + \right.
\\	 & 
+ \left.  630 \left(\frac{12\beta}{\alpha^2}-7 \right)  U U_xU_{xx} + 630 \left( \frac{2\beta}{\alpha^2} - 1\right)   U^2U_{3x} + 210 \left(  \frac{48\gamma}{\alpha^3}    -  \frac{60\beta}{\alpha^2}  + 23  \right)  U^3U_x +  \right. 
\\
& + \left. 210 \langle U^2\rangle \left(U_{3x} + \left(\frac{18\beta}{\alpha^2} - 9 \right)UU_x \right) + 105 \left( 3 \langle U_x^2\rangle + \left(\frac{12 \beta}{\alpha^2}-10 \right) \langle U^3\rangle 
 +6\langle U^2\rangle\langle U \rangle\right) U_x \right\}+\\
&  + \mathcal{O}(h^8)\, .
\end{split}
\end{equation}
\end{theorem}

\begin{proof}
We shall  look for an approximate solution of  the invariance  equation  (\ref{eq:invman2}) of the form 
\begin{equation}
\label{eq:cUexp}
c(U,h)= h^2 c_2(U) + h^4c_4(U)+\mathcal{O}(h^6)\ .
\end{equation}
Combining (\ref{eq:invman2}) with the explicit expression for $\mathscr{F}$
given in \eqref{eq:EvolutionU}, we get
\[
\begin{split}
&\left[h^2 c_2(U) + h^4c_4(U) +\mathcal{O}(h^6)\right]'\left[U_x+h^2\left(
\frac{1}{24}U_{3x}+{\color{black}\frac{1}{8}(U^2)_x}\right)+\mathcal{O}(h^4)\right]=\\
&-\left[h^2c_2(U)+h^4c_4(U)+\mathcal{O}(h^6)\right]_x-h^2\left[\frac{1}{8}{\color{black}(U^2)_x}
+\frac{h^2}{24}(c_2(U))_{3x}+\frac{h^2}{4}(Uc_2)_x+\mathcal{O}(h^4)
\right]\\
& - h^4\left[\frac{1}{192}{\color{black}(U^2)_{3x}}+\frac{\beta}{32\alpha^2}{\color{black}(U^3)_x}\right]
+ \mathcal{O}(h^6)\ .
\end{split}
\]
Collecting terms of orders $h^2$ and $h^4$ produces two equations:
\begin{align}
\label{eq:c2eq}
c_2'(U)U_x & =-\frac{1}{8}UU_x\ ,
\\
\label{eq:c4eq}
c_4'(U)U_x & =-c_2'(U)\left(\frac{1}{24}U_{3x}+\frac{1}{8}{\color{black}(U^2)_x}\right)
-\frac{1}{24}(c_2(U))_{3x}-\frac{1}{4}(Uc_2(U))_x-\frac{1}{192}{\color{black}(U^2)_{3x}}
-\frac{\beta}{32\alpha^2}{\color{black}(U^3)_x}\ .
\end{align}
A solution of equation \eqref{eq:c2eq} is easily found:
\begin{equation}
\label{eq:c2sol}
c_2(U)=\frac{1}{16}\left(\langle U^2 \rangle-U^2\right)\ .
\end{equation}
Here, $\langle U^2 \rangle := \int_{\mathbb T} U^2(x)\, dx$ denotes the average of the function $U^2$. We could have omitted the term $\frac{1}{16}\langle U^2 \rangle$ in \eqref{eq:c2sol}, but including this term in $c_2(U)$ guarantees that $\langle c_2(U)\rangle = 0$. See also Remark \ref{rmk:averages}.

Inserting \eqref{eq:c2sol} into equation \eqref{eq:c4eq} for $c_4(U)$ gives
\[
c_4'(U)U_x=\left(\frac{5}{128}-
\frac{3\beta}{64\alpha^2}\right)U^2U_x-\frac{1}{128}\langle U^2\rangle U_x
-\frac{1}{128}U_{xx}U_x\ ,
\]
which admits as a solution
\begin{equation}
\label{eq:c4sol}
c_4(U)=\left(\frac{5}{384}-
\frac{\beta}{64\alpha^2}\right)U^3-
\frac{1}{128}\langle U^2 \rangle (U -\langle U \rangle)
-\frac{1}{256}[(U_x)^2-\langle (U_x)^2 \rangle]\ .
\end{equation}
Again, we made sure that $\langle c_4(U)\rangle =0$ by choosing an appropriate ``integration constant''. Together, \eqref{eq:cUexp}, \eqref{eq:c2sol}
and \eqref{eq:c4sol} produce \eqref{eq:cU}.

The evolution equation for $U$ is obtained from \eqref{eq:quasiUevolution}  by inserting 
the expansion \eqref{eq:cUexp} into  \eqref{eq:EvolutionU}, giving
\[
\begin{split}
U_t &=\mathcal{F}(U,h):=\mathscr{F}(U,h^2c_2(U)+h^4c_4(U)+\mathcal{O}(h^6),h) = \\
& = U_x+h^2\left(\frac{1}{24}U_{3x}+\frac{1}{8}{\color{black}(U^2)_x}\right)+
 h^4\left[\frac{1}{1920}U_{5x}+
\frac{1}{192}{\color{black}(U^2)_{3x}+\frac{\beta}{32\alpha^2}(U^3)_x}+\frac{1}{4}(Uc_2(U))_x\right] +\\
& + h^6\left[ \frac{1}{322560}U_{7x}+{\color{black}\frac{1}{15360}(U^2)_{5x}+
\frac{\beta}{768\alpha^2}(U^3)_{3x}+\frac{\gamma}{128\alpha^3}(U^4)_x}+\right.\\
& +\left.\frac{1}{8}(c_2(U)^2)_x+\frac{1}{96}(Uc_2(U))_{3x}+
\frac{3\beta}{32\alpha^2}(U^2c_2(U))_x+\frac{1}{4}(Uc_4(U))_x
\right]+\mathcal{O}(h^8)\ .
\end{split}
\] 
Substituting the expression \eqref{eq:c2sol} and \eqref{eq:c4sol} that we found for $c_2(U)$ and $c_4(U)$  
gives \eqref{eq:Fu}.
This completes the proof of the theorem.
\end{proof}
\begin{remark}
\label{rmk:averages}
For any (smooth) function $F=F(x): \mathbb T \to \mathbb R$ we denote by 
$$\langle F\rangle := \int_{\mathbb T} F(x)\, dx\, ,$$
the {\it average} of $F$ over $\mathbb T$. In Theorem \ref{thm:F} we encounter averages of the functions $F=U, U^2, U^3$ and $(U_x)^2$. By including such average terms as ``integration constants'' in our choice for $c_2$ and $c_4$, we make sure that $\langle c_2(U) \rangle =\langle c_4(U) \rangle = 0$. This choice is of course somewhat arbitrary, but it makes that the solutions $c_2$ and $c_4$ of \eqref{eq:c2eq} and \eqref{eq:c4eq} are unique.
\end{remark}

\section{A normal form theorem}
\label{sec:normform}

In \eqref{eq:Riemanninv} we defined the discrete Riemann invariants $U$ and $V$ in terms of the interpolating profiles $u$ and $v$ that were introduced in Section \ref{sec:formulation}. Our choice to define them as $U=2\alpha (D_hu +v)$ and $V=2 \alpha(D_hu - v)$  helped us in obtaining a closed expression for the equations that govern their evolution, see \eqref{eq:exactriemanninv}, but  was otherwise somewhat arbitrary. For example, we could have equally well defined $U$ as $2 \alpha (u_x+v)$ and $V$ as $2 \alpha (u_x-v)$, the difference between the former and the latter choices being only of the order $h^2$.  

Another definition of the variables $U$ and $V$ would have resulted in another expression for the slaving relation $V=c(U,h)$. More importantly, another choice of $U$ would have given another -- and perhaps  
simpler -- expression for the reduced vector field $\mathcal{F}(U,h)$ that determines the dynamics  on the invariant manifold of quasi unidirectional waves, see \eqref{eq:quasiUevolution}. This motivates us to look for a near-to-identity transformation $U\mapsto U + \mathcal O(h^2)$, i.e.,  a small  change in the definition of $U$, that simplifies the evolution equation $U_t = \mathcal{F}(U,h)$ as much as possible. It turns out that we can prove the following normal form theorem.

\begin{theorem}\label{th:Main}
Recall that $\mathcal{F} = \mathcal{F}(U,h): C^{\infty}(\mathbb{T}, \R) \to C^{\infty}(\mathbb{T}, \R)$ admits the asymptotic expansion \eqref{eq:Fu}. There exists a formal change of variables inside the space $C^{\infty}(\mathbb{T}, \R)$ of the form
\begin{equation}\label{eq:NearToIdTransf}
U \mapsto U + h^2 \widetilde G_2(U) +h^4 \widetilde G_4(U) + h^6 \widetilde G_6(U,t) + \mathcal{O}(h^8) \, ,
\end{equation}
 that transforms the evolution equation $U_t = \mathcal{F}(U,h)$ into
\begin{align}
\label{eq:NormalisedThesis41}
\begin{split}
U_t  = \widetilde{\mathcal{F}}(U,h, t)    & {}^{} = \hspace{5mm} {}^{} \mathcal{C}_1 (U,h)\ K_1(U)+  \\ 
 & {}^{}+ h^{2}\ \mathcal{C}_3(U,h) \ K_3(U)+ \\
& {}^{}+ h^{4}\  \mathcal{C}_5(U,h)\ K_5(U)+ \\
& {}^{} + h^{6}\  \mathcal{C}_7(U, h) \left[K_7(U)+R(U)\right] + \mathcal O(h^{8})\, .
\end{split}
\end{align}
\noindent 
Here, $K_1(U), K_3(U), K_5(U), K_7(U)$ are the first four commuting vector fields in the KdV
hierarchy given in Remark \ref{rmk:hierarchy}. 
The scalars $\mathcal{C}_1(U,h),\mathcal{C}_3(U,h), \mathcal{C}_5(U,h), \mathcal{C}_7(U,h)$ are constants of motion of the KdV hierarchy. They are given explicitly in Remark \ref{rmk:constantsFPU} below.
 The term $R=R(U)$ can be chosen equal to zero when $$14 \alpha^3 - 27 \alpha \beta +12 \gamma=0\, .$$
\end{theorem}

\begin{remark} \label{rem:Toda}
A Toda chain is an FPU chain with a potential energy of the form
\begin{equation}\label{eq:TodaPotentialEnergy}
W(z) = W_{{\rm T}}(z):= \frac{e^{2 \alpha z}-(1+2 \alpha z)}{4 \alpha^2} = \frac{1}{2} z^2 + \frac{\alpha}{3} z^3 + \frac{\alpha^2}{6}z^4 + \frac{\alpha^3}{15} z^5 +  \ldots .
\end{equation}
 Toda chains thus define a specific one-parameter family of FPU chains, for which
$$\beta= \beta_{T} := \frac{2}{3} \alpha^2,\ \gamma= \gamma_T:=\frac{1}{3} \alpha^3,\ \mbox{etc.}$$
One readily checks that $14 \alpha^3 - 27 \alpha \beta_T +12 \gamma_T=0$. This means that for the Toda chains, the normal form to order $h^6$ given in Theorem \ref{th:Main} lies in the KdV hierarchy. 
\end{remark}

\noindent Theorem \ref{th:Main} follows from a more general result that we will state below as Theorem \ref{thm:Kodama}. A result comparable to Theorem \ref{thm:Kodama} was originally formulated by Hiraoka and Kodama in \cite{HirKod}. These authors consider evolution equations similar to \eqref{eq:Fu} for functions $U\in C^{\infty}(\mathbb{R}, \mathbb{R})$. In this paper we work with functions $U\in C^{\infty}(\mathbb{T}, \mathbb{R})$. We thus need to adapt the proof of Hiraoka and Kodama. 
\begin{remark}
The main new feature in our adaptation concerns the use of primitives $F_{-x}$ of functions $F: \mathbb T \to \mathbb R$. Primitives of functions $F\in C^{\infty}(\mathbb R, \mathbb R)$ are always defined. In contrast, a function $F\in C^{\infty}(\mathbb T, \mathbb R)$ only possesses a well defined primitive $F_{-x}\in C^{\infty}(\mathbb T, \mathbb R)$ if $\langle F \rangle = \int_{\mathbb T} F(x)\, dx = 0$. 
To deal with this complication we shall only make use of  transformations  in $C^{\infty}(\mathbb T, \mathbb R)$ that map the space of zero-average functions on $\mathbb T$ into itself. This in turn forces us to add some extra ``average'' terms to the  transformations that were originally considered by Hiraoka and Kodama.

In particular, we shall agree that, if $F$ has zero average, then $F_{-x}$ is the unique primitive of $F$ that has zero average itself. For a general $F\in C^{\infty}(\mathbb T, \mathbb R)$, we will then have the formulas
$$\langle (F - \langle F\rangle)_{-x} \rangle = 0 \ \mbox{and}\ ((F-\langle F\rangle)_{-x})_x = F-\langle F\rangle\, .$$
\end{remark}
 
\noindent Examples of transformations   involving primitives occur in Theorem \ref{thm:Kodama} below, which   makes use of maps $G_2, G_4 : C^{\infty}(\mathbb T, \mathbb R) \to C^{\infty}(\mathbb T, \mathbb R)$ given by formulas of the form
\begin{equation}\label{eq:DefG2}
	G_2(U) \;:=\; \frac{C_5}{C_3} \left(a_1 U_{2x} + a_2 (U^2-\langle U^2 \rangle) +a_3 (U_x(U- \langle U\rangle)_{-x} + \langle U^2\rangle - \langle U \rangle^2) +a_4\langle U \rangle (U-\langle U \rangle) \right)
\end{equation}
and
\begin{equation}\label{eq:DefG4}
\begin{split}
	G_4(U)\;&:=\;  \frac{C_7}{C_3}\big(b_1U_{4x} + b_2((U_x)^2-\langle (U_x)^2 \rangle ) + b_3 (UU_{2x}+\langle (U_x)^2\rangle) + b_4(U^3-\langle U^3 \rangle ) \\
	&\qquad + b_5(U_x(U^2-\langle U^2 \rangle)_{-x}+\langle U^3 \rangle - \langle U \rangle \langle U^2 \rangle)\\
	&\qquad +b_6 ((U_{3x} + 6UU_{x})(U-\langle U \rangle)_{-x}+3 \langle U^3 \rangle - \langle (U_x)^2 \rangle -3 \langle U^2 \rangle \langle U \rangle	) \\
	&\qquad +b_7 \langle U \rangle U_{xx} + b_8 \langle U \rangle (U^2-\langle U^2 \rangle)\\
	&\qquad +b_9 \langle U \rangle (U_x(U-\langle U \rangle)_{-x}+ \langle U^2 \rangle - \langle U \rangle^2)+b_{10} \langle U \rangle^2(U-\langle U \rangle)+ b_{11} \langle U^2 \rangle( U-\langle U \rangle)\\
	&\qquad  +b_{12}\langle (U_x)^2 \rangle + b_{13} \langle U^3 \rangle \big) \, .
\end{split}
\end{equation}
Here, $a_1, \ldots, a_4, b_1, \ldots b_{13}, C_3, \ldots, C_7\in \mathbb R$ are constants that will be specified later. Formulas \eqref{eq:DefG2} and \eqref{eq:DefG4} are such that $\langle G_2(U)\rangle = \langle G_4(U)\rangle = 0$ for every $U\in C^{\infty}(\mathbb T, \mathbb R)$. 

The following theorem is our version of the result of Hiraoka and Kodama.
\begin{theorem}[Hiraoka--Kodama on $\mathbb{T}$]\label{thm:Kodama}
Consider an evolution equation for $U\in C^{\infty}(\mathbb T, \mathbb R)$ of the  form
\begin{equation} \label{eq:KodamaEq}
\begin{split}
U_t \;=\;  \mathrm{F}(U,h)  \;=\; & \hspace{3.5mm} C_1 \, U_x  \\
+\; & h^{2}  C_3 \left(U_{3x}+6UU_x\right) \\
+\; & h^{4} C_5 \left(U_{5x} + A_1 U_xU_{2x} + A_2UU_{3x} + A_3U^2U_x  +A_4 \langle U^2 \rangle U_x \right) \\ 
+ \; & h^{6} C_7 \, \big(U_{7x} + B_1U_{2x}U_{3x} + B_2 U_xU_{4x} + B_3 UU_{5x} + B_4 (U_x)^3 + B_5 UU_xU_{2x} +  \\
& \qquad + B_6U^2U_{3x} + B_7 U^3U_x  +\langle U \rangle \big(B_8 U_{5x} +B_9 U_x U_{2x}+B_{10} UU_{3x} +B_{11} U^2 U_x\big) \\
& \qquad +\langle U^2 \rangle (B_{12} U_{3x}+B_{13} U U_x) + \langle U \rangle^2 (B_{14} U_{3x}+ B_{15} U U_x) \\
&\qquad  +(B_{16} \langle U^3 \rangle + B_{17} \langle (U_x)^2 \rangle  + B_{18} \langle U \rangle \langle U^2 \rangle + B_{19} \langle U \rangle^3) U_x + B_{20} \langle (U_x)^3 \rangle \big) \\
+ \; &  \mathcal{O}(h^8)\, .
\end{split}
\end{equation}
Here $A_1, \ldots, A_4$, $B_1, \ldots, B_{20}, C_1, \ldots, C_7$ are scalar coefficients.
\begin{itemize}
\item[i)] By a normal form transformation of the form 
\begin{equation}\label{eq:NormalForm5}
U \mapsto U + h^2 G_2(U) + \mathcal{O}(h^4)\, ,
\end{equation}
with $G_2$ of the form \eqref{eq:DefG2}, one can
 transform equation \eqref{eq:KodamaEq} into the form
\begin{equation}\label{eq:KDVto5}
U_t \; = \; \mathcal{C}_1(U,h) K_1(U) + h^{2}  \mathcal{C}_3(U,h) K_3(U) + h^{4} \mathcal{C}_5(U,h)  K_5(U) + \mathcal{O}(h^{6})\, .
\end{equation}
The scalars $\mathcal{C}_1(U,h), \ldots, \mathcal{C}_5(U,h)$ are  constants of motion of the KdV hierarchy,  explicitly given by
\begin{equation}\label{eq:veryexplicitfirst}
	\begin{split}
		\mathcal{C}_1(U, h) \;&=\; C_1(1+ h^4 (\widetilde{A}_4 \langle U^2 \rangle + \widetilde{A}_5 \langle U \rangle^2) \, , \\
		\mathcal{C}_3(U, h)\;&=\; C_3(1+h^2 \widetilde{A}_6 \langle U \rangle) \, , \\
		\mathcal{C}_5(U, h) \;&=\; C_5  \, .
	\end{split}
\end{equation}
The scalars $\widetilde{A}_4, \widetilde{A}_5, \widetilde{A}_6$ in  \eqref{eq:veryexplicit}  depend on $A_1, A_2, A_3$ and $A_4$ as follows:
\begin{align}
  \label{eq:EquationsForAT}  
\begin{split}
		\widetilde{A}_4 \;&=\; A_3+A_4-4 A_2+10\, , \\
		\widetilde{A}_5 \; &= \; 20-2 A_2 \, , \\
		\widetilde{A}_6 \; &= \; A_2-10 \, . \\
	\end{split}  
\end{align} 
\item[ii)] By a further normal form transformation of the form
\begin{equation}\label{eq:NormalForm7}
U\mapsto \ U + h^4 G_4(U) + \mathcal{O}(h^8)\, ,
\end{equation}
with $G_4$ of the form \eqref{eq:DefG4}, one can subsequently 
 transform  equation \eqref{eq:KDVto5} into the form
 \begin{equation}
     \label{eq:KDVto7}
\begin{split}
U_t\; =\; &  \mathcal{C}_1(U,h) K_1(U) + h^{2} \mathcal{C}_3(U,h) K_3(U) + h^{4} \mathcal{C}_5(U,h) K_5(U) \\ 
 +\; & h^{6} \mathcal{C}_7(U,h) [ K_7(U)+ R(U) +    \lambda_7 \langle (U_x)^3 \rangle  ] +\mathcal O(h^{8})\, .
\end{split}
 \end{equation}
The term $R(U)$ can be chosen equal to zero when $r = r({\bf A},{\bf B},{\bf C})=0$, where
\begin{equation}\label{eq:KodamaRCorrect}
\begin{split}
	r \;:=&\;1680-72 B_1 + 180 B_2- 510 B_3-72 B_3+27 B_5- 72 B_4+27 B_5+24 B_6-9B_7 \\
	&\;+\frac{C_5^2}{C_3 C_7} \left\{-2400+6A_1^2+670 A_2-30 A_1 A_2-4A_2^2-60 A_3+3A_1A_3-A_2A_3 \right\} \, .
	\end{split}
\end{equation}
The scalars $\mathcal{C}_1(U,h), \ldots, \mathcal{C}_7(U,h)$ are  constants of motion of the KdV hierarchy,  explicitly given by  
\begin{equation}\label{eq:veryexplicit}
	\begin{split}
		\mathcal{C}_1(U, h) \;&=\; C_1(1+ h^4 (\widetilde{A}_4 \langle U^2 \rangle + \widetilde{A}_5 \langle U \rangle^2) + h^6 (\lambda_4 \langle (U_x)^2-2 U^3 \rangle + \lambda_5 \langle U \rangle \langle U^2 \rangle+\lambda_6 \langle U \rangle^3)\, , \\
		\mathcal{C}_3(U, h)\;&=\; C_3(1+h^2 \widetilde{A}_6 \langle U \rangle + h^4 ( \lambda_2 \langle U^2 \rangle + \lambda_3 \langle U \rangle^2)) \, , \\
		\mathcal{C}_5(U, h) \;&=\; C_5(1+h^2 \lambda_1 \langle U \rangle) \, , \\
		\mathcal{C}_7(U, h) \: & =\; C_7\, .
	\end{split}
\end{equation}
Here, the  $\widetilde{A}_4, \widetilde{A}_5, \widetilde{A}_6$ are as in \eqref{eq:EquationsForAT}, the scalar $\lambda_4$ is a free parameter, and
\begin{align}
& 	\begin{split}\label{eq:EquationsForRs}
		\lambda_1 \;&=\; -14+B_3+B_8 \, , \\
		\lambda_2 \;&=\; 42+B_{12}-8 B_3+B_6-\frac{2 C_5^2}{3 C_3 C_7}(A_2-10)^2 \, , \\
		\lambda_3 \;&=\; 28+B_{10}+B_{14}-2 B_3-10 B_8+ \frac{2 C_5^2}{3 C_3 C_7}(100-20 A_2+A_2^2) \, , \\
		\lambda_5 \;&=\; 10B_3-2 B_6+10B_8-4 B_{10}+B_{11}-6B_{12}+B_{13}+B_{18} \\
		&\qquad+ \frac{C_5^2}{C_3 C_7}(-100+40 A_2-A_2^2-10A_3-\frac{2}{3} A_2 A_3+\frac{1}{3}A_3^2-10 A_4 +\frac{1}{3} A_3 A_4) \, , \\
		\lambda_6 \;&=\; -56+4 B_3+20 B_8-2 B_{10}-6B_{14}+B_{15}+B_{19}\\
		& \qquad+\frac{2 C_5^2}{3 C_3 C_7}(-100+A_2^2+10 A_3 - A_2 A_3+10 A_4-A_2 A_4) \, , \\
		\lambda_7 \;&=\; \frac{28}{3}+ \frac{B_{16}}{2}+B_{17}+B_{20}- \frac{7 B_3}{3}+\frac{B_6}{3}\\&\qquad +\frac{C_5^2}{18 C_3 C_7}(-300 +70A_2-A_2^2-A_2 A_3-3 A_1A_4+6 A_2 A_4) \, .
	\end{split}
\end{align}
\end{itemize}
\end{theorem}
\noindent
We will postpone the proof of Theorem \ref{thm:Kodama} to the next section, and  conclude the current section by proving Theorem \ref{th:Main} from Theorem \ref{thm:Kodama}.

\begin{proof}[Proof of Theorem \ref{th:Main}]
The vector field $\mathcal{F}(U,h)$ given in equation \eqref{eq:Fu} is of the form of the vector field $\mathrm{F}(U,h)$ in equation \eqref{eq:KodamaEq} with
\begin{equation}\label{eq:ABC}
	\begin{split}
		{\bf A}\;&=\; (A_1, A_2, A_3, A_4) \;=\; \left(60,20,90 \left(\frac{2 \beta }{\alpha ^2}-1\right),30\right) \, ,  \\
		{\bf B}\;&=\;(B_1, \ldots, B_{20}) \;=\; \left(420,210,42,315 \left(\frac{8 \beta }{\alpha ^2}-5\right),630 \left(\frac{12 \beta }{\alpha ^2}-7\right),630 \left(\frac{2 \beta }{\alpha
   ^2}-1\right), \right.\\
   &\qquad\left. 210 \left(\frac{48 \gamma }{\alpha ^3}-\frac{60 \beta }{\alpha ^2}+23\right) ,0,0,0,0,210,210 \left(\frac{18 \beta }{\alpha
   ^2}-9\right),0,0,105 \left(\frac{12 \beta }{\alpha ^2}-10\right),315,0,0,0\right) \, ,
\\
  {\bf C} \; & =\; (C_1, C_3, C_5, C_7) = \left(1, \frac{1}{24},  \frac{1}{1920}, \frac{1}{322560} \right) \, . 
  	\end{split}
\end{equation}
(NB: for clarity, we denote vectors with bold characters.) According to Theorem \ref{thm:Kodama}   it is thus possible to transform \eqref{eq:Fu} into \eqref{eq:KDVto7}   with two consecutive near-to-identity transformations of the form \eqref{eq:NormalForm5} and \eqref{eq:NormalForm7} respectively. To remove the term $h^6 \mathcal{C}_7(U,h) \lambda_7 \langle (U_x)^3\rangle = h^6 C_7 \lambda_7 \langle (U_x)^3\rangle$ in  \eqref{eq:KDVto7}, 
it suffices to apply one more time-dependent change of variables
$$U \mapsto U + h^6 \widetilde G_6(U,t) := U-h^{6} C_7 \lambda_7 \int_0^t \langle U_x^3(s) \rangle \,  \mathrm{d}s \, ,$$ 
which cancels this last term in \eqref{eq:KDVto7}, yielding precisely \eqref{eq:NormalisedThesis41}. 

The quantity that determines whether the term $R(U)$ in \eqref{eq:NormalisedThesis41} and \eqref{eq:KDVto7} can be chosen equal to zero, is the constant $r$ defined in \eqref{eq:KodamaRCorrect}. 
   Substituting  \eqref{eq:ABC} in \eqref{eq:KodamaRCorrect} gives 
$$r = - \frac{7560}{\alpha^3} \left( 14 \alpha^3-27 \alpha \beta +12 \gamma  \right)\, . $$
This finishes the proof of Theorem \ref{th:Main}.
\end{proof}
\begin{remark}\label{rmk:constantsFPU}
Explicit expressions for the constants of motion $\mathcal{C}_1(U,h), \ldots, \mathcal{C}_7(U,h)$ in Theorem \ref{th:Main} can be obtained by combining \eqref{eq:veryexplicit} with  \eqref{eq:EquationsForAT},  \eqref{eq:EquationsForRs} and \eqref{eq:ABC}. In the setting of Theorem \ref{th:Main} we have
\begin{equation*}
	\widetilde{A}_4\;=\;-130+\frac{180 \beta}{\alpha^2} \, , \qquad \widetilde{A}_5 \;=\;-20 \, , \qquad \widetilde{A}_6\;=\;10 \, ,
\end{equation*}
so that
\eqref{eq:veryexplicitfirst} becomes
 \begin{equation}\label{eq:CsFPUfirst}
	\begin{split}
		&\mathcal{C}_1(U,h) = 1+h^4\left(({\textstyle \frac{180 \beta}{\alpha^2}-130) } \langle U^2 \rangle-20 \langle U \rangle^2\right), \,
		\mathcal{C}_3(U,h) = \frac{1}{24} \left(1+10 \, h^2 \langle U \rangle  \right), \, \mathcal{C}_5(U,h) = \frac{1}{1920}  \,   .
	\end{split}
\end{equation}
By combining \eqref{eq:EquationsForRs}  and \eqref{eq:ABC}   we obtain
\begin{equation*}
	\begin{split}
		&\lambda_1\;=\;28 \, , \qquad \lambda_2 \;=\; -854+ \frac{1260 \beta}{\alpha^2} \, , \qquad \lambda_3\;=\;84 \, , \\
		&\lambda_5 \;=\;420\left(16-\frac{63 \beta}{\alpha^2}+\frac{54 \beta^2}{\alpha^4} \right) \, , \qquad \lambda_6\;=\;28\left(49-\frac{90 \beta}{\alpha^2}\right) \, , \qquad \lambda_7\;=\;-427+\frac{630 \beta}{\alpha^2} \, .
	\end{split}
\end{equation*}
Recall that the scalar $\lambda_4$ may be chosen freely, so let us choose $\lambda_4=0$. 
Then
\eqref{eq:veryexplicit} becomes
 \begin{equation}\label{eq:CsFPU}
	\begin{split}
		&\mathcal{C}_1(U,h) \;=\; 1+h^4\left(({\textstyle \frac{180 \beta}{\alpha^2}-130) } \langle U^2 \rangle-20 \langle U \rangle^2\right)+ \\
		& \hspace{1.75cm} + h^6\left({\textstyle 420(16-\frac{63 \beta}{\alpha^2}+\frac{54 \beta^2}{\alpha^4}) \langle U \rangle \langle U^2 \rangle+28(49-\frac{90 \beta}{\alpha^2}) \langle U \rangle^3}\right)\, , \\
		&\mathcal{C}_3(U,h)\;=\;\frac{1}{24} \left(1+10 \, h^2 \langle U \rangle + h^4\left({\textstyle (\frac{1260 \beta}{\alpha^2}-854)\langle U^2 \rangle+84 \langle U \rangle^2}\right)\right)\, , \\
		&\mathcal{C}_5(U,h)\;=\; \frac{1}{1920}\left(1+28 \, h^2 \langle U \rangle \right) \, ,\\
		&\mathcal{C}_7(U,h)=\frac{1}{322560}\, .
	\end{split}
\end{equation}
\end{remark}
\section{Proof of Theorem \ref{thm:Kodama}}\label{sec:KodamaProof}
In the proof of Theorem \ref{thm:Kodama} that we provide in this section, we shall  explicitly compute how coordinate changes of the form $U \mapsto U+h^2 G_2(U)+ \mathcal{O}(h^4)$ and $U\mapsto U + h^4 G_4(U) + \mathcal{O}(h^8)$, with $G_2(U)$ given by \eqref{eq:DefG2} and $G_4(U)$ by \eqref{eq:DefG4}, transform an evolution equation $U_t={\rm F}(U,h)$ of the form \eqref{eq:KodamaEq}. In particular, we shall compute when exactly it is possible to transform \eqref{eq:KodamaEq} into a member of the KdV hierarchy, to order $h^4$ and to order $h^6$, and which choices of $G_2(U)$ and $G_4(U)$ realise this transformation.

Our proof is divided in three steps: first  we transform equation \eqref{eq:KodamaEq} into a normal form to order $h^{4}$. Next, we compute how this first transformation affects the evolution equation to order $h^{6}$. And finally, we normalise this new equation  to order  $h^{6}$.
 
Before we give a more concise outline of this procedure, we recall the definition of the Lie bracket of two operators $f$ and $g$:
\begin{equation} \label{eq:bracketgeneral}
	[f,g](U)\; :=\;f'(U)g(U) - g'(U)f(U)\, .
\end{equation}
Here $f'(U)g(U)$ denotes the Gateaux derivative (or directional derivative) of an operator $f$ (evaluated at $U$) in the direction of $g(U)$. See also Remark \ref{rmk:gateaux}. More explicitly, 
$$f'(U)g(U) \; :=\; \left. \frac{d}{d\varepsilon}\right|_{\varepsilon =0} f(U+\varepsilon g(U))\, .$$
In this paper, $f$ and $g$ will always be operators from $C^{\infty}(\mathbb T, \mathbb R)$ to $C^{\infty}(\mathbb T, \mathbb R)$.  

\paragraph{Outline of the normal form procedure.}
Below we sketch the procedure by which we bring equation \eqref{eq:KodamaEq} in Theorem \ref{thm:Kodama} into normal form. For other settings in which normal form transformations are applied to the FPUT chain, we refer to \cite{BP06, Gall, HK, Bob}. Recall that equation \eqref{eq:KodamaEq} is of the form 
\begin{equation}
U_t = \mathrm{F}_1(U) + h^{2}\mathrm{F}_3(U) + h^{4} \mathrm{F}_5(U) + h^{6}\mathrm{F}_7(U) + \mathcal{O}(h^{8})\, ,
\end{equation}
in which 
\begin{equation} \label{eq:Fs}
\begin{split}
    \mathrm{F}_1(U) &= C_1U_x \, , \\
    \mathrm{F}_3(U) & = C_3 (U_{3x}+6U U_x)\, , \\
    \mathrm{F}_5(U) & = C_5\left(U_{5x} + A_1 U_x U_{2x} + A_2 U U_{3x} + A_3U^2U_x + A_4\langle U^2\rangle U_x \right)\, , \\
    \mathrm{F}_7(U) & = C_7 \big(U_{7x} + B_1U_{2x}U_{3x} + B_2 U_xU_{4x} + B_3 UU_{5x} + B_4 (U_x)^3 + B_5 UU_xU_{2x} + B_6U^2U_{3x}  \\
& + B_7 U^3U_x +\langle U \rangle \big(B_8 U_{5x} +B_9 U_x U_{2x}+B_{10} UU_{3x} +B_{11} U^2 U_x\big) \\
& +\langle U^2 \rangle (B_{12} U_{3x}+B_{13} U U_x) + \langle U \rangle^2 (B_{14} U_{3x}+ B_{15} U U_x) \\
&  +(B_{16} \langle U^3 \rangle + B_{17} \langle (U_x)^2 \rangle  + B_{18} \langle U \rangle \langle U^2 \rangle + B_{19} \langle U \rangle^3) U_x + B_{20} \langle (U_x)^3 \rangle \big)\, .
\end{split}
\end{equation}

\noindent
We will transform this equation by a normal form procedure that consists of two separate transformation steps.  First we make a   transformation of the form $U\mapsto e^{h^2 G_2}(U)  := U+h^2G_2(U) + \frac{h^4}{2}G_2'(U)G_2(U)+ \mathcal{O}(h^6)$ where $G_2$ is chosen in such a way that $[G_2, \mathrm{F}_1]=0$. This transforms the vector field ${\rm F}={\rm F}(U)$ into  $$e^{h^2[G_2,\cdot]}({\rm F}) :=1+h^2[G_2,{\rm F}]+ \frac{1}{2} h^4 [ G_2, [G_2, {\rm F}]] + \frac{1}{6} h^6 [G_2, [ G_2, [G_2, { \rm F}]]] + \mathcal O (h^8)\, .$$ 
Expanding  ${\rm F}$ in powers of $h$ in this equation, we  obtain the expansion 
\begin{equation}\label{eq:NormalFormTransfGen}
\begin{split}
    U_t  = & \mathrm{F}_1(U) + h^{2}\mathrm{F}_3(U) + \\ + & 
    h^{4} \big\{  \underbrace{  \mathrm{F}_5(U)+ [G_2, \mathrm{F}_3](U)  }_{=: {\rm N}_5(U)} \big\} +   h^{6}\big\{  \mathrm{F}_7(U) + \underbrace{  [G_2,\mathrm{F}_5](U) + \frac{1}{2}[G_2, [G_2, \mathrm{F}_3]](U)}_{=\frac{1}{2}[G_2, {\rm F}_5+{\rm N}_5](U) =: {\rm R}_6(U)}\big\}  + \mathcal{O}(h^{8}) \, ,
    \end{split}
    \end{equation}
   for the transformed evolution equation. Note that the term $h^2 {\rm F}_3(U)$ is unaffected because $[G_2, {\rm F}_1]=0$. We now want 
  to choose $G_2$ in such a way that ${\rm N}_5 = \mathrm{F}_5+[G_2, \mathrm{F}_3]$ is in the KdV hierarchy, see Remark \ref{rmk:hierarchy}. It turns out that this can always be arranged. 
    
    The next step is to make a further coordinate change 
     $U\mapsto e^{h^4 G_4}(U)  = U+h^4G_4(U) + \mathcal{O}(h^8)$, again choosing $G_4$ so that $[G_4, \mathrm{F}_1]=0$. This   transforms our evolution equation further into 
    \begin{equation}\label{eq:NormalFormTransfGen2}
    U_t = \mathrm{F}_1(U) + h^{2}\mathrm{F}_3(U) + h^{4} {\rm N}_5(U) + h^{6}\big\{ \underbrace{   \mathrm{F}_7(U) + \frac{1}{2}[G_2, \mathrm{F}_5 + {\rm N}_5](U) + [G_4, \mathrm{F}_3](U)}_{= {\rm F}_7(U) + {\rm R}_6(U) + [G_4, {\rm F}_3](U)=:{\rm N}_7(U)} \big\}  + \mathcal{O}(h^{8}) \, .
    \end{equation}
    The goal is to choose $G_4$ so that ${\rm N}_7=\mathrm{F}_7 + \frac{1}{2}[G_2, \mathrm{F}_5 + {\rm N}_5] + [G_4, \mathrm{F}_3]$ is in the KdV hierarchy as well. It turns out that this can only be arranged (within the class of transformations that we consider) if  a certain relation among the coefficients of ${\rm F_5}$ and ${\rm F}_7$ is satisfied.
    
\paragraph{Normalisation at lowest order.} 
Inspired by Lemma 5.1 in \cite{HirKod} we choose $G_2$
of the form \eqref{eq:DefG2}. We have that $[G_2, {\rm F_1}] = C_1 [G_2, \partial_x] = 0$ because $G_2=G_2(U)$ does not explicitly depend on $x$. 

We also remark that  the authors of \cite{HirKod} use  a simpler class of transformations of the form $(C_3/C_5) G_2(U) = a_1 U_{2x} + a_2 U^2 + a_3 U_x U_{-x}$. This is not sufficient for us, on the one hand because equation \eqref{eq:KodamaEq} that we try to bring into normal form, is more general than the equations considered in \cite{HirKod}.  On the other hand, we like to make sure that  $\langle G_2(U)\rangle = 0$, so   that $e^{h^2G_2}$ maps the space of zero-average functions into itself.   
To compute the transformed vector field ${\rm N}_5(U)$, we use the bracket relations in Table \ref{tab:Table1}.
\begin{table}[]
    \centering
$$\begin{array}{|l|l|}
\hline
X &  [X, U_{3x} + 6 UU_x,]\\ \hline \hline 
U_{2x}   & 12 U_xU_{2x}  \\ \hline
U^2-\langle U^2 \rangle   & -6 U_xU_{2x} - 6U^2U_x +  6 \langle U^2\rangle U_x   \\ \hline
U_x (U-\langle U\rangle)_{-x}+\langle U^2\rangle - \langle U\rangle^2  & -3 U_xU_{2x}-3 UU_{3x}  - 3U^2U_x - 9  \langle U^2\rangle U_x + 6 \langle U\rangle^2U_x +6\langle U\rangle UU_x + 3 \langle U\rangle U_{3x}  \\ \hline
\langle U\rangle (U - \langle U \rangle)  & -6 \langle U \rangle U U_x +6 \langle U \rangle^2 U_x  \\ 
\hline
 \end{array}
$$
    \caption{Lie brackets determining the first normalisation step.}
    \label{tab:Table1}
\end{table}
\begin{remark}
The brackets in Table \ref{tab:Table1} can all be computed by hand, using formula \eqref{eq:bracketgeneral}. To illustrate,
\begin{equation*}
\begin{split}
     \left[ U^2 - \langle U^2\rangle, U_{3x} + 6 U U_x \right] \; = \; &  \!
     \left. \frac{d}{d\varepsilon}\right|_{\varepsilon=0} \big\{  
    \left( U + \varepsilon \left(U_{3x}+6U U_x \right)\right)^2 - \langle \left( U + \varepsilon \left(U_{3x}+6U U_x\right) \right)^2 \rangle  \\ & 
     - \left( U + \varepsilon \left( U^2 - \langle U^2\rangle \right)\right)_{3x} - 6\left( U + \varepsilon \left( U^2 - \langle U^2\rangle \right)\right) \left( U + \varepsilon \left( U^2 - \langle U^2\rangle \right)\right)_x \big\}  \\
     \; = \; & 2 U (U_{3x} + 6 U U_x) - \langle 2 U (U_{3x} + 6 U U_x) \rangle  \\ 
     & - \left( U^2 - \langle U^2\rangle\right)_{3x} - 6 U \left( U^2 - \langle U^2\rangle \right)_x - 6 \left( U^2 - \langle U^2\rangle \right) U_x\, .
    \end{split}
\end{equation*}

Now observe that
$\langle UU_{3x} \rangle = - \langle U_xU_{2x} \rangle = - \langle (\frac{1}{2} (U_x)^2)_x \rangle = 0$ (using integration by parts and the fundamental theorem of calculus), that $\langle U^2U_{x} \rangle  = \langle (\frac{1}{3} U^3)_x \rangle = 0$ (again by the fundamental theorem of calculus), and that 
$\langle U^2\rangle_{3x} = \langle U^2\rangle_{x} = 0$ (because the average is not a function of $x$).
 Using these identities, 
as well as the chain rule to rewrite the terms $(U^2)_{3x}$ and $U(U^2)_x$, we can simplify our expression for the bracket to $-6 U_xU_{2x} - 6U^2U_x +  6 \langle U^2\rangle U_x$.   
The other brackets are computed in  an analogous fashion.
\end{remark}

\noindent Using \eqref{eq:DefG2}, \eqref{eq:Fs} and the table above, we compute that 
\begin{equation}\label{eq:TransformedFirstOrder}
\begin{split}
 [G_2, \mathrm{F}_3](U)   \;=\;  & [G_2, C_3(U_{3x}+6UU_x)]  \\
  \;=\; & C_5\left\{ (12 a_1 -6a_2-3a_3 )U_xU_{2x}  - 3a_3 UU_{3x} - (6a_2+3a_3) U^2U_x + \right. \\
& \left.  (6a_3 - 6 a_4)\langle U \rangle UU_x +  3a_3 \langle U \rangle U_{3x} +(6a_2-9a_3)\langle U^2\rangle U_x + (6a_3+6a_4)\langle U\rangle^2 U_x \right\} \, .
\end{split}
\end{equation}
Recall that we want to make our lowest order normal form a member of the KdV hierarchy, of the form

\begin{equation}
\label{eq:wantKdV5}
\begin{split}
    N_5(U) & =\mathrm{F}_5(U)+[G_2, \mathrm{F}_3](U) = \\ & C_5(U_{5x}+20U_xU_{2x} + 10 U U_{3x}+ 30 U^2 U_x + \widetilde{A_4} \langle U^2\rangle U_x +\widetilde{A_5} \langle U\rangle^2 U_x + \widetilde{A_6} \langle U \rangle (U_{3x} + 6 UU_x) ) \, .
    \end{split}
\end{equation}
 
This can be arranged if we can solve the  system of linear equations
\[
\left\{
\begin{array}{l}   A_1 + 12 a_1 -6a_2-3a_3 = 20 \, , \\
 A_2 -3a_3 = 10\, ,\\
A_3- 6a_2 - 3a_3 = 30\, ,\\
6a_3 - 6a_4 = 18 a_3\, . 
\end{array} \right.
\]
for the coefficients $a_1,\dots,a_4$ that define $G_2$. It turns out that this system admits the unique solution 
\begin{equation}\label{eq:CoefficientsA}
\begin{split}
    a_1 &= \frac{1}{12}(A_3 - A_1 - 10) \, , \\
    a_2 &= \frac{1}{6}(A_3 - A_2 - 20) \, , \\
    a_3 &= \frac{1}{3}(A_2 - 10) \, , \\
    a_4 &= \frac{2}{3}(10-A_2)\, .
\end{split}
\end{equation}
Inserting \eqref{eq:CoefficientsA}  back into \eqref{eq:TransformedFirstOrder} and comparing with \eqref{eq:wantKdV5}, we in fact obtain
\[
\begin{split}
\widetilde{A}_4 &= A_4 + 6a_2-9a_3 = A_3+ A_4 -4A_2 + 10 \, , \\
\widetilde{A}_5 &= 6a_3+6a_4 = 20-2A_2 \, ,\\
\widetilde{A}_6 &= 3 a_3  = A_2-10\, .
\end{split}
\]
This coincides with \eqref{eq:EquationsForAT}.  Equations \eqref{eq:veryexplicitfirst} follow from \eqref{eq:wantKdV5} combined with \eqref{eq:EquationsForAT}. 
This completes the proof of part {\it i)} of Theorem \ref{thm:Kodama}. 

\paragraph{A new higher order term.}
The normal form transformation generated by $h^2G_2(U)$ simplifies our evolution equation  to order $h^4$, but it adds  new additional terms at order $h^{6}$. From \eqref{eq:NormalFormTransfGen} we see that these new terms are  
\[
 {\rm R}_6(U)\;=\; \frac{1}{2}[G_2, \mathrm{F}_5+{\rm N}_5](U) \,  .
\]
Here $G_2$ is as determined as in the previous paragraph, and we in fact have
\[
\begin{split}
  \mathrm{F}_5(U)+N_5(U) = & C_5\left \{ 2 U_{5x}+(A_1+20)U_xU_{2x} + (A_2+10) U U_{3x}+ (A_3+30) U^2 U_x + \right. \\ & \left. (A_3+2A_4-4A_2+10) \langle U^2\rangle U_x +(20-2A_2) \langle U\rangle^2 U_x + (A_2-10) \langle U \rangle (U_{3x} + 6 UU_x) \right\} \,.
    \end{split}
\]
To compute ${\rm R}_6(U)$ explicitly, we use  the bracket relations listed in Table \ref{tab:Table2}.
\begin{table}[]
    \centering
\[
\begin{array}{|l|l|l|}
\hline
X & Y & [X,Y]\\ \hline \hline 
U_{xx} &  U_{5x} & 0 \\
 \hline
  " &   U_{x}U_{2x} & 2   U_{2x}U_{3x}\\
  \hline
  " &   UU_{3x} & 2  U_{x}U_{4x} \\
   \hline
  " &  U^2U_{x} &  2(U_x)^3+4 UU_xU_{2x} \\ \hline
  " & \langle U^2\rangle U_x & 2\langle (U_x)^2\rangle U_x 
  \\ \hline
  " & \langle U\rangle^2 U_x &  0  \\
  \hline
   " & \langle U\rangle (U_{3x} + 6 U U_x) &  12 \langle U\rangle U_xU_{2x}  \\
  \hline
U^2-\langle U^2 \rangle & U_{5x} & -20 U_{2x}U_{3x} - 10 U_xU_{4x} \\
 \hline
  " &  U_{x}U_{2x} & -  2 (U_x)^3-2UU_xU_{xx}+\langle (U_x)^3\rangle \\
  \hline
  " &   UU_{3x} & -6UU_xU_{2x}-U^2U_{3x} -2\langle (U_x)^3\rangle +\langle U^2\rangle U_{3x}\\ \hline
  " & U^2U_{x} &-2U^3U_x  +2\langle U^2\rangle UU_x \\
  \hline
   " & \langle U^2\rangle U_x &  -2\langle U^3\rangle U_x + 2 \langle U \rangle \langle U^2 \rangle U_x\\
   \hline
     " & \langle U\rangle^2 U_x &  0  \\
  \hline
   " & \langle U\rangle (U_{3x} + 6 U U_x) &   -6\langle U \rangle U_xU_{2x} -6\langle U \rangle U^2U_x + 6 \langle U \rangle \langle U^2\rangle U_x     \\
  \hline
   U_x(U-\langle U\rangle)_{-x} +\langle U^2\rangle - \langle U\rangle^2 & U_{5x} & -15U_{2x}U_{3x}-10U_xU_{4x}-5UU_{5x} +5 \langle U \rangle U_{5x}\\
 \hline
  " &  U_{x}U_{2x} & -\frac{1}{2} U_x^3 - 3 UU_xU_{2x} -\frac{1}{2}\langle U_x^2\rangle U_x -\langle U_x^3\rangle + 3\langle U\rangle U_xU_{2x} \\
  \hline
  " &  UU_{3x} & -\frac{1}{2}(U_x)^3-3UU_xU_{2x}-3U^2U_{3x} +\frac{3}{2}\langle (U_x)^2\rangle U_x + \\
  & & 2 \langle (U_x)^3\rangle -\langle U^2\rangle U_{3x} +\langle U\rangle^2 U_{3x} + 3\langle U \rangle UU_{3x}   \\
   \hline
  " &   U^2U_{x} & -\frac{2}{3}U^3U_x -\frac{1}{3}\langle U^3\rangle U_x -2\langle U^2\rangle UU_x + 2\langle U\rangle^2 UU_x + \langle U\rangle U^2U_x\\ \hline
   " & \langle U^2\rangle U_x & \langle U^3\rangle U_x - 3 \langle U \rangle \langle U^2\rangle U_x +  2 \langle U\rangle^3 U_x  \\
  \hline
     " & \langle U\rangle^2 U_x &   0    \\
  \hline
   " & \langle U\rangle (U_{3x} + 6 U U_x) &  - 3 \langle U\rangle U_x U_{2x} - 3 \langle U\rangle U U_{3x}  - 3 \langle U\rangle U^2 U_{x}  \\
   & &   + 3\langle U\rangle^2 U_{3x}  + 6\langle U\rangle^2 UU_{x} + 6 \langle U\rangle^3 U_{x} - 9 \langle U\rangle\langle U^2\rangle U_{x}\\
    \hline
   \langle U\rangle (U-\langle U \rangle) & U_{5x} & 0 \\
 \hline
  " &  U_xU_{2x} & -\langle U \rangle U_x U_{2x} \\
  \hline
  " & U U_{3x} & \langle U \rangle^2 U_{3x} -\langle U \rangle U U_{3x}\\
   \hline
  " & U^2 U_x &  2\langle U\rangle^2 UU_x- 2\langle U\rangle U^2U_x\\ \hline
   " & \langle U^2\rangle U_x & 2 \langle U \rangle^3 U_x - 2 \langle U \rangle \langle U^2\rangle U_x\\
  \hline
     " & \langle U\rangle^2 U_x &  0  \\
  \hline
   " & \langle U\rangle (U_{3x} + 6 U U_x) &  6\langle U\rangle^3 U_x - 6 \langle U\rangle^2UU_x     \\
   \hline
\end{array}
\]
 \caption{Lie brackets determining the new higher order term. }
    \label{tab:Table2}
\end{table}
Using Table \ref{tab:Table2}, it is straightforward to compute that 

\begin{equation}\label{eq:ThirdOrderBeforeNormalisation}
\begin{split}
{\rm R}_6(U)\;=\;
\frac{C_5^2}{C_3}
\Big( & \widetilde{B}_1 U_{2x}U_{3x} + \widetilde{B}_2 U_xU_{4x} + \widetilde{B}_3 UU_{5x} + \widetilde{B}_4 (U_x)^3 + \widetilde{B}_5 UU_xU_{2x} + \widetilde{B}_6U^2U_{3x} + \widetilde{B}_7 U^3U_x + \\
& \langle U \rangle \left( \widetilde{B}_{8} U_{5x} + \widetilde{B}_{9}  U_xU_{2x} + \widetilde{B}_{10}  UU_{3x}  +  \widetilde{B}_{11}  U^2U_{x} \right) + \\
&  \langle U^2 \rangle \left( \widetilde{B}_{12} U_{3x} +  \widetilde{B}_{13} U U_x \right)  + \langle U\rangle^2 \left( \widetilde{B}_{14}  U_{3x}  + \widetilde{B}_{15} U U_x \right) +\\
&\left(
\widetilde{B}_{16} \langle U^3 \rangle + 
\widetilde{B}_{17}  \langle (U_x)^2 \rangle  +  
\widetilde{B}_{18} \langle U\rangle \langle U^2\rangle + \widetilde{B}_{19} \langle U \rangle^3  \right) 
U_x  + 
\widetilde{B}_{20} \langle (U_x)^3\rangle  \Big) \, ,
\end{split}
\end{equation} 
in which 
\begin{equation} \label{eq:tildeB}
 \begin{array}{l}
 
\widetilde{B}_1 =  a_1(A_1 + 20)- 20a_2-15a_3 \, ,  \\  
 \widetilde{B}_2 = a_1(A_2+10) -10a_2-10a_3 \, , \\  
 \widetilde{B}_3 =  -5a_3\, , \\  
 \widetilde{B}_4 =    a_1(A_3+30) -a_2(A_1+20)-(a_3/4)(A_1+A_2 + 30)\, , \\  
 \widetilde{B}_5 =   2a_1(A_3+30)-a_2(A_1+3A_2+50)-(3a_3/2)(A_1+A_2+30)\, , \\  
 \widetilde{B}_6 =  -(a_2+3a_3)(A_2+10)/2 \, , \\  
 \widetilde{B}_7 =   -(a_2+a_3/3)(A_3+30) \, , \\  
 \widetilde{B}_8 = 5a_3 \, ,
  \\  
 \widetilde{B}_9 = (6a_1-3a_2)(A_2-10) +3a_3(A_1-A_2+30)/2 - a_4(A_1+20)/2 \, ,
  \\  
 \widetilde{B}_{10} = 30a_3 -a_4(A_2+10)/2\, ,
  \\  
 \widetilde{B}_{11} = (6a_2+3a_3)(10-A_2)/2+(a_3 -2a_4)(A_3+30)/2\, ,
  \\  
 \widetilde{B}_{12} = (a_2-a_3)(A_2+10)/2\, ,
  \\  
 \widetilde{B}_{13} = (a_2-a_3)(A_3+30)\, ,
  \\  
 \widetilde{B}_{14} =  a_3(2A_2-10)+a_4(A_2+10)/2\, ,
  \\  
 \widetilde{B}_{15} =  (a_3+a_4)(A_3+30) +3(a_3-a_4)(A_2-10)\, ,
  \\  
 \widetilde{B}_{16} =  (a_3/2-a_2)(A_3+2A_4-4A_2+10) -a_3(A_3+30)/6 \, ,
  \\  
 \widetilde{B}_{17} = a_1(A_3+2A_4-4A_2+10) + a_3(3A_2-A_1+10)/4\, ,
  \\  
 \widetilde{B}_{18} =  (a_2- \frac{3}{2}a_3-a_4)(A_3+2A_4-4A_2+10) + (3a_2-\frac{9}{2}a_3)(A_2-10)\, ,
  \\  
 \widetilde{B}_{19} = (a_3+a_4)(A_3+2A_4-A_2-20)\, , \\  
 \widetilde{B}_{20} = (a_2-a_3)(A_1-2A_2)/2  \, .
\end{array}
\end{equation}
Here the coefficients $a_1$, $a_2$, $a_3$, $a_4$ are  as defined in \eqref{eq:CoefficientsA}. Note that $\widetilde{B}_4 - \frac{1}{2} \widetilde{B}_5 + \widetilde{B}_6 + \widetilde{B}_{20} = 0$ because ${\rm R}_6(U)$ has zero average.
To summarise, we now have that 
$$\mathrm{F}_7(U) + {\rm R}_6(U) = C_7\left( U_{7x} + \left(B_1+\frac{C_5^2}{C_3C_7}\widetilde{B}_1\right) U_{2x}U_{3x} + \left(B_2+\frac{C_5^2}{C_3C_7}\widetilde{B}_2\right) U_{x}U_{4x}  + \ldots \right) $$
is the new term of order $h^6$, after the first normal form transformation.

\paragraph{Normalisation at second order.}
In this final paragraph we perform a second change of variables $U\mapsto e^{h^2 G_4}(U) = U + h^4 G_4(U) + \mathcal O(h^8)$ to transform our evolution equation into an as simple as possible   form at order $h^{6}$.
In \eqref{eq:NormalFormTransfGen2} we see that in this way we add the term $[G_4, {\rm F_3}]$ to the order $h^6$ part of the equation. Recall that  the transformation generator $G_4(U)$ has the specific form \eqref{eq:DefG4}. 

To see what exactly the additional term $[G_4, {\rm F_3}]$ will look like, we  present the relevant  Lie brackets in Table \ref{tab:Table3}.
\begin{table}[]
    \centering
\[
\begin{array}{|l|l|}
\hline
X & [X, U_{3x}+6UU_x] \\
\hline \hline
U_{4x} & 60 U_{2x}U_{3x} + 24U_xU_{4x}  \\
\hline

(U_x)^2- \langle (U_x)^2\rangle & -6U_{2x}U_{3x} + 6 (U_x)^3 -6\langle (U_x)^3\rangle + 6 \langle (U_x)^2\rangle U_x\\ 
\hline
 UU_{2x}+ \langle (U_{x})^2\rangle & -3U_{2x}U_{3x} -3U_xU_{4x} + 12UU_xU_{2x} +6 \langle (U_{x})^3\rangle -6\langle (U_x)^2\rangle U_x \\
\hline
U^3-\langle U^3\rangle & -6(U_x)^3-18UU_xU_{2x} - 6U^3U_x-3 \langle (U_x)^3\rangle +6\langle U^3\rangle U_x\\
\hline
 U_x(U^2-\langle U^2\rangle)_{-x} +\langle U^3\rangle - \langle U\rangle \langle U^2\rangle &  -3(U_x)^3 -6UU_xU_{2x} -3U^2U_{3x} - 2U^3U_x +3\langle U^2\rangle U_{3x} +3 \langle (U_x)^2\rangle U_x  \\ 
&  + 6 \langle U^2\rangle UU_x - 10\langle U^3\rangle U_x  +3\langle (U_x)^3\rangle   + 6 \langle U\rangle \langle U^2\rangle U_x
 \\ \hline
  (U_{3x}+6UU_x)(U- \langle U\rangle)_{-x}  &    -3U_x U_{4x}- 3 UU_{5x} - 18(U_x)^3- 72UU_xU_{2x} - 21U^2U_{3x} - 18U^3U_x    \\
 +3\langle U^3\rangle - \langle (U_x)^2\rangle -  3 \langle U^2\rangle \langle U\rangle  &    - 3\langle U^2\rangle U_{3x} +6\langle (U_x)^2\rangle U_x+ 3\langle (U_x)^3\rangle -18\langle U^3\rangle U_x -18\langle U^2\rangle UU_x 
\\
  &   + 18 \langle U \rangle \langle U^2\rangle U_x + 3\langle U\rangle (U_{5x} + 18U_xU_{2x}+8UU_{3x} + 12U^2U_x )
\\ \hline
 \langle U \rangle U_{xx} &  12 \langle U \rangle U_xU_{2x}  \\
 \hline
 \langle U \rangle ( U^2-\langle U^2\rangle ) &  -6 \langle U \rangle U_xU_{2x} - 6 \langle U \rangle U^2U_x +  6 \langle U \rangle  \langle U^2\rangle U_x \\
 \hline
 \langle U \rangle (U_x (U-\langle U\rangle)_{-x}+\langle U^2\rangle - \langle U\rangle^2 ) & -3 \langle U\rangle U_xU_{2x}-3 \langle U\rangle UU_{3x}  - 3\langle U\rangle U^2U_x - 9 \langle U\rangle \langle U^2\rangle U_x \\ &  + 6 \langle U\rangle^3U_x +6\langle U\rangle^2 UU_x + 3 \langle U\rangle^2 U_{3x} \\ \hline
  \langle U \rangle^2(U-\langle U\rangle) & -6 \langle U \rangle^2 U U_x +6 \langle U \rangle^3 U_x  \\
 \hline
 \langle U^2 \rangle (U-\langle U\rangle) & 6\langle U^2\rangle \langle U\rangle U_x-6 \langle U^2 \rangle UU_x\\ \hline
\langle (U_x)^2 \rangle &  -6 \langle (U_x)^2 \rangle U_x +6 \langle (U_x)^3 \rangle  \\ \hline
\langle U^3 \rangle &  -6 \langle U^3 \rangle U_x +3 \langle (U_x)^3 \rangle \\ \hline
\end{array} 
\]
    \caption{Lie brackets determining the second normalisation step.}
    \label{tab:Table3}
\end{table}

From now on we shall write 
$$\mathbf{B}=(B_1,\dots,B_{20}) \, , \,  \mathbf{\widetilde{B}}=(\widetilde{B}_1,\dots,\widetilde{B}_{20})\, , \,  \mathbf{\widetilde{\widetilde{B}}}=(\widetilde{\widetilde{B}}_1, \dots, \widetilde{\widetilde{B}}_{20})\, , \, \mbox{and} \ \mathbf{b}=(b_1,\dots, b_{13})\, . $$ 
(Again, we denote vectors with bold characters.) Here, the $B_i$ are as in \eqref{eq:Fs}, the $\widetilde B_j$ as in \eqref{eq:tildeB}, and the $b_k$ as in \eqref{eq:DefG4}.   
From the bracket relations in Table \ref{tab:Table3}, it   follows that the new order $h^6$ part of our evolution equation reads
\[
\begin{split}
N_7(U)& :=\mathrm{F}_7(U) + {\rm R}_6(U) + [G_4, \mathrm{F}_3](U) 
= \\
& C_7\left( U_{7x} + \widetilde{\widetilde{B}}_1 U_{2x} U_{3x} +\widetilde{ \widetilde{B}}_2 U_xU_{4x} +\widetilde{ \widetilde{B}}_3 UU_{5x} +\widetilde{ \widetilde{B}}_4 (U_x)^3 +\widetilde{ \widetilde{B}}_5 UU_xU_{2x} +\widetilde{ \widetilde{B}}_6U^2U_{3x} +\widetilde{ \widetilde{B}}_7 U^3U_x + \ldots \right)
\end{split}
\]
-- we use the notation as in \eqref{eq:Fs} -- in which 
\[
\widetilde{\widetilde{\mathbf{B}}} = \mathbf{B} + \frac{C_5^2}{C_3C_7}\widetilde{\mathbf{B}} + M \mathbf{b}\, ,
\]
and $M$ is the $20\times 13$ matrix 
\[
M=\left( \begin{array}{ccccccccccccc}
60 & -6 & -3 & 0& 0& 0 & 0 & 0 & 0 & 0 & 0 & 0 & 0\\
24 & 0 & -3 & 0& 0& -3 & 0 & 0 & 0 & 0 & 0& 0 & 0\\ 
0 & 0 & 0 & 0 & 0 & -3 & 0 & 0 & 0 & 0 & 0 & 0 & 0\\
0 & 6 & 0 & -6 & -3 & -18& 0 & 0 & 0 & 0 & 0 & 0 & 0\\
0 & 0 & 12 & -18 & -6 & -72& 0 & 0 & 0 & 0 & 0 & 0 & 0\\
0 & 0& 0& 0& -3& -21& 0 & 0 & 0 & 0 & 0 & 0 & 0\\
0 & 0 & 0 & -6 & -2 & -18& 0 & 0 & 0 & 0 & 0& 0 & 0\\
0 & 0 & 0 & 0 & 0 & 3 & 0 & 0 & 0 & 0 & 0& 0 & 0\\
0 & 0 & 0 & 0 & 0 & 54 & 12 & -6 & -3 & 0 & 0& 0 & 0\\
0 & 0 & 0 & 0 & 0 & 24 & 0 & 0 & -3 & 0 & 0& 0 & 0\\
0 & 0 & 0 & 0 & 0 & 36 & 0 & -6 & -3 & 0 & 0& 0 & 0\\
0 & 0 & 0 & 0 & 3 & -3 & 0 & 0 & 0 & 0 & 0& 0 & 0\\
0 & 0 & 0 & 0 & 6 & -18 & 0 & 0 & 0 & 0 & -6& 0 & 0\\
0 & 0 & 0 & 0 & 0 & 0 & 0 & 0 & 3 & 0 & 0& 0 & 0\\
0 & 0 & 0 & 0 & 0 & 0 & 0 & 0 & 6 & -6 & 0& 0 & 0\\
0 & 0 & 0 & 6 & -10 & -18 & 0 & 0 & 0 & 0 & 0& 0 & -6\\
0 & 6 & -6 & 0 & 3 & 6 & 0 & 0 & 0 & 0 & 0& -6 & 0\\
0 & 0 & 0 & 0 & 6 & 18 & 0 & 6 & -9 & 0 & 6& 0 & 0\\
0 & 0 & 0 & 0 & 0 & 0 & 0 & 0 & 6 & 6 & 0& 0 & 0\\
0 & -6 & 6 & -3 & 3 & 3 & 0 & 0 & 0 & 0 & 0 & 6 & 3
\end{array}
\right) \, .
\]
Recall that we would like $N_7(U)$ to lie in the KdV hierarchy. More precisely, we want it to be of the form
\begin{equation}\label{eq:wantKdV7}
\begin{split}
& C_7\left(U_{7x} + 70U_{2x}U_{3x} + \ldots \right) \\  
& +  \lambda_1 \langle U\rangle (U_{5x} + 20 U_{x}U_{2x} + \ldots ) +  
\lambda_2 \langle U^2 \rangle (U_{3x} + 6 UU_x)  +  \lambda_3 \langle U \rangle^2 (U_{3x} + 6 UU_x)   \\
& + \lambda_4 (\langle (U_x)^2\rangle -2 \langle U^3\rangle  ) U_x + \lambda_5 \langle U\rangle \langle U^2\rangle U_x +\lambda_6 \langle U\rangle^3 U_x +\lambda_7 \langle (U_x)^3 \rangle\, ,
\end{split} 
\end{equation}
for certain scalars $\lambda_1, \ldots, \lambda_7$.  This can be arranged precisely when the  system of algebraic equations
\begin{equation}\label{eq:solvability}
	 \mathbf{\widetilde{\widetilde{B}}} = \mathbf{B} + \frac{C_5^2}{C_3C_7}\widetilde{\mathbf{B}} + M \mathbf{b} \;=\; \mathbf{w} \, 
\end{equation}
can be solved for ${\bf b}$, where $\mathbf{w} = \mathbf{w}(\lambda)$ is defined as
\begin{equation}
	\mathbf{w}\;=\;(70, 42, 14, 70, 280, 70, 140, \lambda_1, 20 \lambda_1, 10 \lambda_1, 30 \lambda_1, \lambda_2, 6 \lambda_2, \lambda_3, 
 6 \lambda_3, -2 \lambda_4, \lambda_4, \lambda_5, \lambda_6, \lambda_7) \, .
\end{equation}
In fact, the  system \eqref{eq:solvability} admits a solution ${\bf b}$ if and only if 
$$ \mathbf{w} - \mathbf{\widetilde{\widetilde{B}}}  = \mathbf{w}-\frac{C_5^2}{C_3 C_7} \mathbf{\widetilde{B}}-\mathbf{B} \in \mathrm{ran} \, M\, .$$  
Now note that the orthogonal complement of  $\mathrm{ran} \, M$  (with respect to the Euclidean inner product in $\mathbb{R}^{20}$) is $(\mathrm{ran} \, M)^ \perp \;= \; \mathrm{span} \,\{\mathbf{v}_1,\mathbf{v}_2,\mathbf{v}_3,\mathbf{v}_4,\mathbf{v}_5,\mathbf{v}_6,\mathbf{v}_7\}$ with
\begin{equation}
	\begin{split}
		&\mathbf{v}_1\;=\;(0, 0, -14, 0, 0, 2, 0, 0, 0, 0, 0, 0, 0, 0, 0, 3, 6, 0, 0, 6) \, ,  \\
		&\mathbf{v}_2\;=\;(0, 0, 32, 0, 0, 0, 0, 0, 0, 4, 0, 0, 0, 0, 1, 0, 0, 0, 1, 0) \, ,\\
		&\mathbf{v}_3\;=\;(0, 0, -48, 0, 0, 4, 0, 0, 0, -4, 1, 0, 1, 0, 0, 0, 0, 1, 0, 0) \, , \\
		&\mathbf{v}_4\;=\;(0, 0, 8, 0, 0, 0, 0, 0, 0, 1, 0, 0, 0, 1, 0, 0, 0, 0, 0, 0) \, , \\
		&\mathbf{v}_5\;=\;(0, 0, -8, 0, 0, 1, 0, 0, 0, 0, 0, 1, 0, 0, 0, 0, 0, 0, 0, 0) \, , \\
		&\mathbf{v}_6\;=\;(0, 0, 1, 0, 0, 0, 0, 1, 0, 0, 0, 0, 0, 0, 0, 0, 0, 0, 0, 0) \, ,  \\
		&\mathbf{v}_7\;=\;(24, -60, 170, 24, -9, -8, 3, 0, 0, 0, 0, 0, 0, 0, 0, 0, 0, 0, 0, 0) \, .
	\end{split}
\end{equation}
Indeed, one may check that the ${\bf v}_i$ span the kernel of the map ${\bf v}\mapsto {\bf v}^TM$. It is thus sufficient to require that $ \mathbf{w} - \mathbf{\widetilde{\widetilde{B}}}$ is perpendicular to each of the vectors $\mathbf{v}_1, \dots, \mathbf{v}_7$, i.e., that
\begin{equation}\label{eq:scalarproducts}
	\left( \mathbf{B}+ \frac{C_5^2}{C_3 C_7} \mathbf{\widetilde{B}} - \mathbf{w}\right) \cdot \mathbf{v}_j  \;=\; 0\qquad \mbox{for}\ j=1,\dots, 7\, . 
\end{equation}
After substituting  \eqref{eq:CoefficientsA}  in \eqref{eq:tildeB} to obtain an expression for $\widetilde{\bf B}$ in terms of ${\bf A}$ and ${\bf B}$, equations \eqref{eq:scalarproducts} for $j=1,\dots,6$ become precisely equations \eqref{eq:EquationsForRs}. Equation \eqref{eq:scalarproducts} for $j=7$ is equivalent to the equation $r({\bf A},{\bf B},{\bf C})=0$, with $r$ as defined in \eqref{eq:KodamaRCorrect}. This last equation thus constitutes the (only) constraint on the parameters $A_1,\dots,B_{20}$ of the vector fields $\mathrm{F}_5$ and $\rm F_7$. Equations \eqref{eq:veryexplicit}  follow from \eqref{eq:wantKdV5} combined with \eqref{eq:EquationsForAT} and  \eqref{eq:wantKdV7} combined with  \eqref{eq:EquationsForRs}. 
This concludes the proof of part {\it ii)} of Theorem \ref{thm:Kodama}.

{\color{black}
\section*{Acknowledgements}
M.G. was partially supported by the MIUR-PRIN 2017 project MaQuMA cod. 2017ASFLJR, by GNFM (INdAM). M.G. and A.P. acknowledge the kind hospitality of Vrije Universiteit Amsterdam where part of this work was carried out. In turn, B.R. acknowledges the hospitality of the University of Padova. 
A.P. thanks D. Bambusi for many fruitful discussions.
}

\end{document}